\documentclass[a4paper,11pt]{article}
\usepackage{jheppub} 
\usepackage{lineno}
\usepackage{amsmath}
\usepackage{amssymb}
\usepackage{amsfonts}
\usepackage{amsthm}
\usepackage{amscd}
\usepackage{xcolor}
\usepackage{tikz}
\usepackage{appendix}
\usepackage{xcolor}
\usetikzlibrary{shapes, arrows, positioning} 

\tikzstyle{startstop} = [rectangle, rounded corners, minimum width=3cm, minimum height=1cm, text centered, draw=black, fill=red!30]
\tikzstyle{process} = [rectangle, minimum width=3cm, minimum height=1cm, text centered, draw=black, fill=orange!30]
\tikzstyle{decision} = [diamond, minimum width=3cm, minimum height=1cm, text centered, draw=black, fill=green!30]
\tikzstyle{arrow} = [thick,->,>=stealth]

\newcommand{\ew}{\mathcal{E}}%
\newcommand{\cw}{\mathcal{C}}%
\newcommand{\bb}{\partial M}%
\newcommand{\NN}{\mathcal{N}}%
\newcommand{\UU}{\mathcal{U}}%
\newcommand{\ZZ}{\mathcal{Z}}
\newcommand{\CC}{\mathcal{C}}

\newcommand{\RR}{\mathcal{R}}
\newcommand{\ewV}{\mathcal{E}(V_1\cup \cdots \cup V_n)}%
\newcommand{\ewW}{\mathcal{E}(W_1\cup \cdots \cup W_n)}%
\newcommand{\ewX}{\mathcal{E}(X_1\cup \cdots \cup X_n)}%
\newcommand{\ewY}{\mathcal{E}(Y_1\cup \cdots \cup Y_n)}%
\newcommand{\GMI}{\tilde{\Gamma}}
\newcommand{\GMay}{\Gamma_{2\to \text{all}}}
\newcommand{\Vall}{V_1\cup \cdots \cup V_n}

\newcommand{\Xall}{X_1\cup \cdots \cup X_n}

\newcommand{\ZZi}{\mathcal{Z}_{in}}

\newcounter{Counter}
\newtheorem{definition}{Definition}[section]
\newtheorem{lemma}[definition]{Lemma}

\newtheorem{theorem}[definition]{Theorem}

\newtheorem{corollary}[definition]{Corollary}

\newtheorem{remark}[definition]{Remark}
\newtheorem{assumption}{Assumption}

\title{\boldmath New conditions for multipartite entanglement wedge connectivity in $n$-to-$n$ holographic scattering}

\author{Bowen Zhao}
\affiliation{Beijing Institute of Mathematical Sciences and Applications, Beijing, China}

\emailAdd{bowenzhao@bimsa.cn}

\abstract{
We investigate the geometry of entanglement wedges for asymptotic $n$-to-$n$  scattering configurations in asymptotically AdS$_3$ spacetimes.
Extending the $2$-to-$2$ Connected Wedge Theorem, we establish a strictly  weaker sufficient condition for the input entanglement wedge $\mathcal{E}(V_1\cup\cdots\cup V_n)$ to be connected: the existence of a single pair of input regions satisfying a $2$-to-all causal intersection condition already forces full multipartite wedge connectivity.
We also derive novel necessary conditions, showing that when the input wedge is connected, certain output ridges must enter the input wedge, and we organize these consequences into a layered reduction on the boundary lattice.
Furthermore, we analyze the generalized bulk scattering region  $\mathcal{S}_E = \mathcal{E}(V_1\cup\cdots\cup V_n)\cap \mathcal{E}(W_1\cup\cdots\cup W_n)$ and obtain necessary conditions for it to be nonempty; for $n>2$ these conditions are stronger than mere wedge connectedness. 
Our results provide new geometric restrictions on multipartite entanglement in holography and clarify the holographic dictionary for multi‑partite scattering processes, while also highlighting intrinsic limitations for $n>2$.
}

\keywords{AdS/CFT Correspondence,
          Holographic Entanglement Entropy,
          Entanglement Wedge,
          Scattering Amplitudes,
          Superadditivity}

\begin{document}
\maketitle
\flushbottom

\section{Introduction}

The AdS/CFT correspondence posits a duality between a quantum gravity theory in an
asymptotically Anti-de Sitter (AdS) spacetime $M$ and a conformal field theory (CFT) on its
timelike boundary $\partial M$ \cite{maldacena1999AdSCFT,witten1998AdScft}.
A foundational requirement of the correspondence is the consistency of causal structure
between the bulk and the boundary.
Gao and Wald proved that, assuming the null curvature condition and global hyperbolicity,
bulk causality cannot violate boundary causality: if two boundary points are connected by a
causal curve through the bulk, they are also connectable by a causal curve restricted to the
boundary \cite{gao2000theorems}.

A more subtle consistency requirement emerges for asymptotic quantum tasks involving
multiple boundary regions.
Consider an asymptotic $n$-to-$n$ scattering configuration on the boundary $\bb$, specified
by $n$ disjoint input regions $V_1,\cdots,V_n$ and $n$ disjoint output regions
$W_1,\cdots,W_n$.
Local scattering processes can occur in the bulk that have no direct boundary counterpart,
referred to as \emph{bulk-only scattering}
\cite{gary2009bulkonly,heemskerk2009bulkonly,penedones2011bulkonly,maldacena2017bulkpointsingularity}.

For the $2$-to-$2$ case ($n=2$), the Connected Wedge Theorem (CWT)
\cite{maypenington2020holographic} establishes that for such a bulk-only process, the
associated boundary input regions $V_1$ and $V_2$ must share $O(1/G_N)$ mutual information,
$I(V_1:V_2)\sim O(1/G_N)$.
This implies that a local bulk scattering process necessitates nonlocal boundary protocols.

Via the Hubeny--Rangamani--Ryu--Takayanagi (HRRT) prescription
\cite{RT2006formula,HRT2007covariant}, this large mutual information admits a geometric
interpretation: the entanglement wedge of $V_1\cup V_2$ becomes connected.
Under standard assumptions (AdS hyperbolicity, the null curvature condition, and the
maximin construction for HRRT surfaces), the CWT may be stated geometrically
\cite{maypenington2020holographic}:

\begin{theorem}\label{thm:CWT}
Under standard assumptions, for a $2$-to-$2$ bulk-only scattering configuration, the entanglement wedge of $V_1\cup V_2$ is connected.
\end{theorem}

Recent works have elevated the Connected Wedge Theorem to a precise equivalence \cite{zhao2025proof,lima2025sufficientGCWT,caminiti2024holographic}.
They show that the existence of $O(1/G_N)$ mutual information between the two input regions $V_1$ and $V_2$ is equivalent to the non-emptiness of a generalized bulk scattering region, denoted $\tilde{S}_E$.
This region is defined as the intersection of the entanglement wedge of the union of input regions (with the wedges of individual input regions removed) and the entanglement wedge of the union of output regions (with the wedges of individual output regions removed).
This provides a complete geometric characterization of quantum nonlocal scattering in the $2$-to-$2$ case. We refer the reader to \cite{zhao2025proof} for a comprehensive review.

In this paper, we extend these analyses to general asymptotic $n$-to-$n$ scattering processes with $n>2$, working exclusively in asymptotically $\mathrm{AdS}_3$ spacetimes whose dual conformal field theory lives on the 
cylindrical boundary $\mathbb{R}\times S^1$.

A sufficient condition for the connectedness of the input entanglement wedge $\ewV:=\ew(V_1\cup\cdots\cup V_n)$ was previously derived in
Ref.~\cite{may2022nton}:

\begin{theorem}\label{thm:n-n-CWT}
Under standard assumptions, if the $2$-to-all graph $\Gamma_{2\to \mathrm{all}}$ is
connected, then the entanglement wedge $\ewV$ is connected.
\end{theorem}

The $2$-to-all graph $\Gamma_{2\to \mathrm{all}}$ has vertices
$\{1,\dots,n\}$ corresponding to the input regions $V_1,\dots,V_n$.
An edge between vertices $i$ and $j$ is inserted if
\begin{equation}
    J^+[\ew(V_i)] \cap J^+[\ew(V_j)] \cap \bigcap_{k=1}^n J^-[\ew(W_k)] \neq \emptyset .
\end{equation}
We note that Ref.~\cite{caminiti2026geodesics} recently constructed explicit examples showing that this theorem fails to have a converse even in the vacuum, indicating that the multipartite case is intrinsically more constrained.

Our first contribution is to establish a strictly weaker necessary condition for the connectedness of $\ewV$ than Theorem~\ref{thm:n-n-CWT}.
Specifically, we show that the existence of a \emph{single} suitable pair of input regions already suffices:

\begin{theorem}\label{thm:weaker-necessary}
Assume the standard conditions listed in Assumption~\ref{assumption:1}.
If there exists a pair $i\neq j$ such that
\begin{equation}\label{eq:intro_weaker_necessary}
    J^+[\ew(V_i)] \cap J^+[\ew(V_j)] \cap \bigcap_{k=1}^{n} J^-[\ew(W_k)] \neq \emptyset ,
\end{equation}
then the entanglement wedge $\ewV$ is connected.
\end{theorem}

The proof reveals an even weaker condition than \eqref{eq:intro_weaker_necessary}, albeit one with a less transparent physical interpretation.

We then derive new, independent \emph{necessary} conditions for the connectedness of multipartite entanglement wedges.
In particular, we show that when both the input wedge $\ewV$ and the output wedge $\ewW:=\ew(W_1\cup\cdots\cup W_n)$ are connected, one necessarily finds a pair of enlarged output regions whose entanglement wedge intersects $\ewV$.
A complete statement is given in Theorem~\ref{thm:sufficient_n}.

Finally, motivated by the central role of the entanglement wedge intersection
$\mathcal{S}_E=\ew(V_1\cup V_2)\cap \ew(W_1\cup W_2)$ in the $2$-to-$2$ case, we analyze its
generalization to $n$-to-$n$ processes.
We provide a necessary condition for $\mathcal{S}_E=\ewV\cap \ewW \neq\emptyset$
(Theorem~\ref{thm:S_E_condition}).
Our results indicate that for $n>2$, the existence of a nontrivial $\mathcal{S}_E$ is governed
by constraints stronger than simple connectedness of the input and output wedges, reflecting
the intrinsically multipartite nature of the problem.

The paper is organized as follows.
Section~\ref{sec:boundary_setup_n} reviews the geometric setup for asymptotic $n$-to-$n$
scattering on the boundary $\bb$.
Section~\ref{sec:Causal_anchoring} recalls the causal anchoring principle.
Section~\ref{sec:ridge_relation} summarizes key geometric results concerning intersections
of bulk causal boundaries and introduces the geometric ordering of ridges.
Section~\ref{sec:characterize_connect} reviews criteria for connectedness and
disconnectedness of multipartite entanglement wedges.
Our results on necessary and sufficient conditions for the connectedness of $\ewV$ are
presented in Sections~\ref{sec:may2022} and~\ref{sec:connect_cons}, respectively.
Section~\ref{sec:S_E_condition} discusses conditions under which the generalized scattering
region $\mathcal{S}_E$ is nonempty.
We conclude with a discussion in Section~\ref{sec:discussion}.
In particular, Section~\ref{sec:comparison_Null_sheet} compares different null-sheet
constructions used in the literature.

\subsection{Notations and Assumptions}\label{subsec:notation_assumptions}
Here we summarize the notations, conventions, and assumptions used throughout this paper.

We adopt natural units with $\hbar=c=1$ and set the AdS length scale $l_{\text{AdS}}=1$, while keeping Newton's constant $G_N$ explicit. Our notation follows ref. \cite{waldGR}, using the mostly-plus metric signature.

\begin{itemize}
\item \textbf{Spacetime regions:} Bulk regions are denoted by script letters ($\mathcal{U}, \mathcal{V}, \mathcal{W}, \cdots$), while boundary regions use straight capitals ($U, V, W, \cdots$). The same symbol may denote either a causal diamond or its Cauchy surface, with the meaning clear from context.

\item \textbf{Cauchy slices:} Bulk Cauchy slices are denoted by $\Sigma$ with appropriate subscripts, boundary Cauchy slices by $\hat{\Sigma}$ with subscripts. By abuse of notation, $\Sigma$ may also refer to Cauchy slices of the conformally compactified spacetime.

\item \textbf{Causal structure:} The bulk causal future/past of region $\mathcal{V}$ is $J^\pm[\mathcal{V}]$; for boundary region $V$, we write $J^\pm[V]$ for bulk causal influence and $\hat{J}^\pm[V]$ for boundary causal influence.

\item \textbf{Domains of dependence:} The bulk domain of dependence of $\mathcal{V}$ is $\mathcal{D}[\mathcal{V}]$; the boundary domain of dependence of $V$ is $\hat{D}[V]$. The future and past horizons of a causal domain $V$ is $\hat{H}^\pm[V]$.

\item \textbf{Entanglement structures:} For boundary region $V$, we denote the entanglement wedge by $\ew(V)$, causal wedge by $\cw(V)$, and HRRT surface by $\text{RT}(V)$.

\item \textbf{Complements:} The causal complement (bulk or boundary) uses superscript $c$, while set-theoretic complement within a Cauchy slice uses superscript prime notation ($'$).
\end{itemize}

\begin{assumption}\label{assumption:1}
We assume throughout that:
\begin{enumerate}
\item The bulk spacetime $M$ satisfies the null curvature condition;
\item HRRT surfaces can be found via a maximin procedure;
\item The spacetime is AdS-hyperbolic (the conformal compactification $\overline{M}=M \cup \partial M$ admits a Cauchy slice);
\item The spacetime region between some Cauchy slice preceding $\ew(V_1\cup V_2)$ and some Cauchy slice following $\ew(W_1\cup W_2)$ is singularity-free \footnote{This assumption excludes the formation of black hole horizons in the relevant part of the spacetime; our results therefore apply to geometries where no horizon has formed by the time of the scattering process.}.
\item The global boundary state is pure, ensuring that a boundary region $V$ and its causal complement $V'$ share the same HRRT surface.
\end{enumerate}
\end{assumption}

\section{Review and Preliminary}\label{sec:review}

\subsection{Boundary setup of n-to-n scattering}\label{sec:boundary_setup_n}
The set-up of $n$-to-$n$ asymptotic scattering is discussed in detail in ref. \cite{may2022nton}. We summarize the setup here with a slightly different formulation.

The boundary configuration for the $n$-to-$n$ scattering process consists of input points $c_1$, $c_2$, ..., $c_n$ and output points $r_1$, $r_2$, ..., $r_n$. Let $\hat{\Sigma}_1$ be a boundary spacelike Cauchy slice containing all $c_i$'s and let $\hat{\Sigma}_2$ be a boundary spacelike Cauchy slice containing all $r_j$'s. A case of $n=3$ is shown in Figure \ref{fig:setup-3} for illustration. 
\begin{figure}
    \centering
    \includegraphics[width=0.8\linewidth,trim={5 11.5cm 5 5},clip]{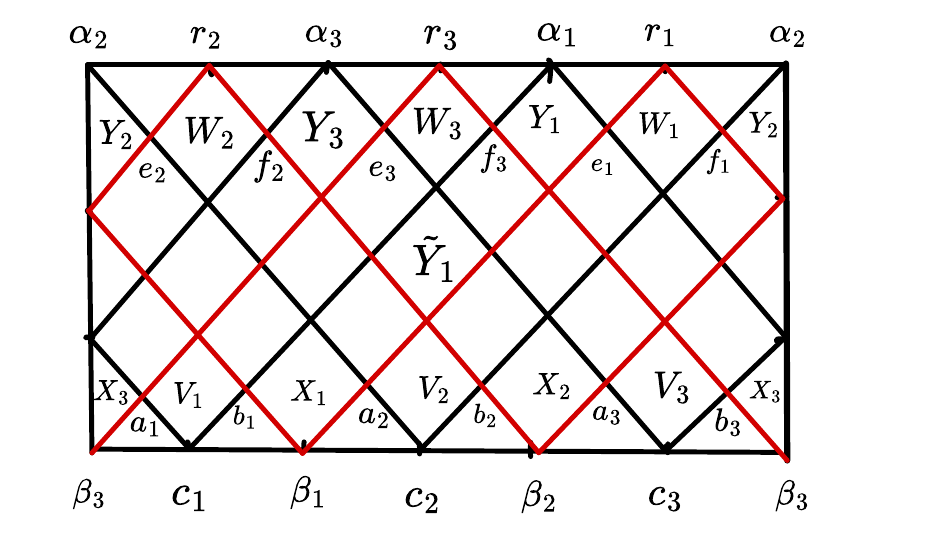}
    \caption{Boundary set-up of $3$-to-$3$ scattering process. Points $c_1, \cdots, c_n$ denote inputs while points $r_1,\cdots, r_n$ denote outputs. The point $\alpha_i$ is the conjugate point of $c_i$ while the point $\beta_j$ is the conjugate point of $r_j$. Input regions $V_i$, with spacelike boundary points $a_i$ and $b_I$, and output regions $W_j$, with spacelike boundary points $e_j$ and $f_j$, are also shown. The causal domain $\tilde{Y}_1$ is marked for later reference.}
    \label{fig:setup-3}
\end{figure}

Recall that the input/decision regions and output regions are defined as
\begin{align*}
    V_i=\hat{J}^+[c_i] \cap \hat{J}^-[r_1]\cap \cdots \cap \hat{J}^-[r_n],\\
    W_i=\hat{J}^-[r_i] \cap \hat{J}^+[c_1]\cap \cdots \cap \hat{J}^+[c_n]
\end{align*}
which are all non-empty sets on  $\bb$ by construction. That is, each input $c_i$ can causally signal all outputs and each output $r_i$ can be causally signaled by all inputs $c_i$. Meanwhile, we require pairwise intersection among these input and output regions to be empty, i.e. 
\begin{align}\label{eq:boundary_condition}
    V_i\cap V_j= \emptyset, \quad W_i \cap W_j &= \emptyset, \quad \forall \, i\neq j \nonumber \\
    V_i \cap W_j &=\emptyset, \quad \forall \, i,j
\end{align}
That is, we require $2$-to-$n$ and $n$-to-$2$ scattering regions to be empty on $\bb$.

We will show that these requirements force the null rays from $c_i$'s and $r_j$'s to form a lattice on $\bb$. To explain this, we label future antipodal points of $c_i$ by $\alpha_i$ and past antipodal points of $r_j$ by $\beta_j$. For example, $\alpha_1$ is the future antipodal point of $c_1$ on $\partial M$. That is, the two future-directed null geodesics emanating from $c_1$ converge at $\alpha_1$. Similarly, the two past-directed null geodesics from $r_j$ converge at $\beta_j$. 

To start with, we label $c_i$ and $r_j$ such that the number increases to the right, or $c_1, \cdots, c_n$ and $r_1, \cdots, r_n$ are ordered counterclockwise when viewed from the future. Since the boundary is topologically $S^1\times \mathbb{R}^1$, we always count modulo $n$; that is, $1=n+1 \mod n$, $0=n \mod n$ and $-1=n-1  \mod n$ etc. All indices are henceforth understood modulo n.
We choose an arbitrary input point to be $c_1$, and the labels of all other input points then follow from the ordering. We still have the freedom to choose which output point is $r_1$. 

Since $c_1$ could causally signal all output points, its future antipodal point $\alpha_1$ must lie between two \textit{adjacent} output points. We can use the freedom of labeling $r_1$ to choose $r_1$ to be the output point to the right of $\alpha_1$. Then, it follows that $\alpha_i$ must lie between $r_{i-1}$ and $r_i$ for all $i\in \{1,\cdots,n\}$. As a result, we have $\alpha_1, r_1, \cdots \alpha_n, r_n$ cyclically ordered (counterclockwise when viewed from future direction) on $\hat{\Sigma}_2$ (we can choose $\hat{\Sigma}_2$ to also contain all $\alpha_i$'s.).

Similarly, since $r_j$ can be causally signaled by all $c_i$, its past antipodal point $\beta_j$ must lie between two \textit{adjacent} input points. It is not difficult to see that $\beta_j$ is forced to lie between $c_j$ and $c_{j+1}$ (Figure \ref{fig:setup-3}). Therefore, on $\hat{\Sigma}_1$ (chosen to also contain all $\beta_j$'s), we have $c_1, \beta_1, \cdots, c_n, \beta_n$ in cyclic order, whose future light rays to the right and to the left form a coordinate lattice on $\bb$. These light rays are also past light rays to the left and to the right, respectively, from $\alpha_1, r_1, \cdots \alpha_n, r_n$.

Figure \ref{fig:setup-3} summarizes the setup for $n=3$. Since we use the flat metric for the conformal boundary $\bb$ as usual, one can trust one's intuition in generalizing Figure \ref{fig:setup-3} to general $n$. 

We also label $X_j$ and $Y_i$ associated to $\beta_j$ and $\alpha_i$. That is,
\begin{align}
    X_i&=\hat{J}^+[\beta_i] \cap \hat{J}^-[\alpha_{1}] \cap \cdots \cap \hat{J}^-[\alpha_n] , \label{eq:X_defn} \\
    Y_i&=\hat{J}^-[\alpha_i] \cap \hat{J}^+[\beta_1] \cap \cdots \cap \hat{J}^+[\beta_n]. \label{eq:Y_defn} 
\end{align}
Since $c_i$ and $\alpha_i$ are antipodal to each other, $V_i$ and $Y_i$ will show up together in following analysis. Similar is true for $X_j$ and $W_j$.
For later convenience, we also label spacelike boundaries of $V_i$, following \cite{may2022nton}. Let $a_i$ be the common boundary between $V_i$ and $X_{i-1}$ and $b_i$ be the common boundary between $V_i$ and $X_{i+1}$. Let $e_i$ be the common boundary between $W_i$ and $Y_i$ and $f_i$ be the common boundary between $W_i$ and $Y_{i+1}$. We note that the relative labelling of $c_i$ and $\alpha_i$ differs from that in ref. \cite{zhao2025proof} (the $\alpha_2$ there would be $\alpha_1$ here).

\subsection{Causal Anchoring Principle}\label{sec:Causal_anchoring}
We recall a crucial observation made in ref. \cite{zhao2025proof}.
The Gao-Wald Theorem implies that for a boundary causal domain $V = \hat{J}^-[p] \cap \hat{J}^+[q]$, the bulk causal wedge is $J^+[p] \cap J^-[q]$. Taking $c_1$ as an example, the null sheet $\partial J^+[c_1]$, which contains the causal surface of $V_1$, is anchored at $\partial \hat{J}^+[c_1]$.

Theorems in ref. \cite{headrick2014causality} generalize the Gao-Wald Theorem to homology regions:
\begin{align}
        \ew(V)\cap \bb &= \hat{D}(V) \label{eq:ew_bb}, \\
    J^{\pm}[RT(V)]\cap \bb &= \hat{J}^{\pm}[\partial V] \label{eq:RT_bb}.
\end{align}
Specifically, null sheets emanating from HRRT surfaces of a causal domain $V$ are anchored at $\hat{J}^\pm[\partial V]$ on $\partial M$. The same is true for null sheets emanating from causal surfaces, due to the causal wedge-entanglement wedge inclusion relation.

In asymptotically global AdS spacetimes, matter/curvature distorts bulk null sheets $\mathcal{N}$ relative to their pure AdS counterparts $\mathcal{N}'$, but their boundary restrictions $\mathcal{N}\cap \partial M$ coincide by \eqref{eq:ew_bb} and \eqref{eq:RT_bb}. Our proof strategy therefore uses boundary null rays from relevant points to constrain the bulk geometry of entanglement wedges and causal wedges.

\subsection{Intersections among wedge horizons}\label{sec:ridge_relation}
Our main proofs rely extensively on geometric relations among null sheets emanating from HRRT surfaces. We therefore summarize some key observations here.

\begin{lemma}\label{lemma:2_null_ridge}
Let $c_1,c_2$ be two distinct points on a boundary Cauchy slice of the timelike boundary $\bb$. 
Then the intersection of their boundary causal futures consists of two points,
\[
\hat{J}^+[c_1]\cap \hat{J}^+[c_2] = \{p,q\}.
\]
Consider two bulk \emph{causal boundaries} $\NN_1$ and $\NN_2$ satisfying
\[
\NN_1\cap \bb=\hat{J}^+[c_1], \qquad \NN_2\cap \bb=\hat{J}^+[c_2].
\]
Then the intersection
\[
\RR := \NN_1 \cap \NN_2
\]
is a connected, continuous, everywhere spacelike, simple (non-self-intersecting) curve with endpoints
$p$ and $q$ on $\bb$.  Moreover, $\RR$ lies entirely in the bulk except at its endpoints.
We refer to such an intersection curve as a \emph{ridge}.

Equivalently, the union $\NN_1\cup\NN_2$ partitions the spacetime into four connected regions.
\end{lemma}

\begin{figure}
    \centering
    \includegraphics[width=0.3\linewidth,trim={10 5 10 5},clip]{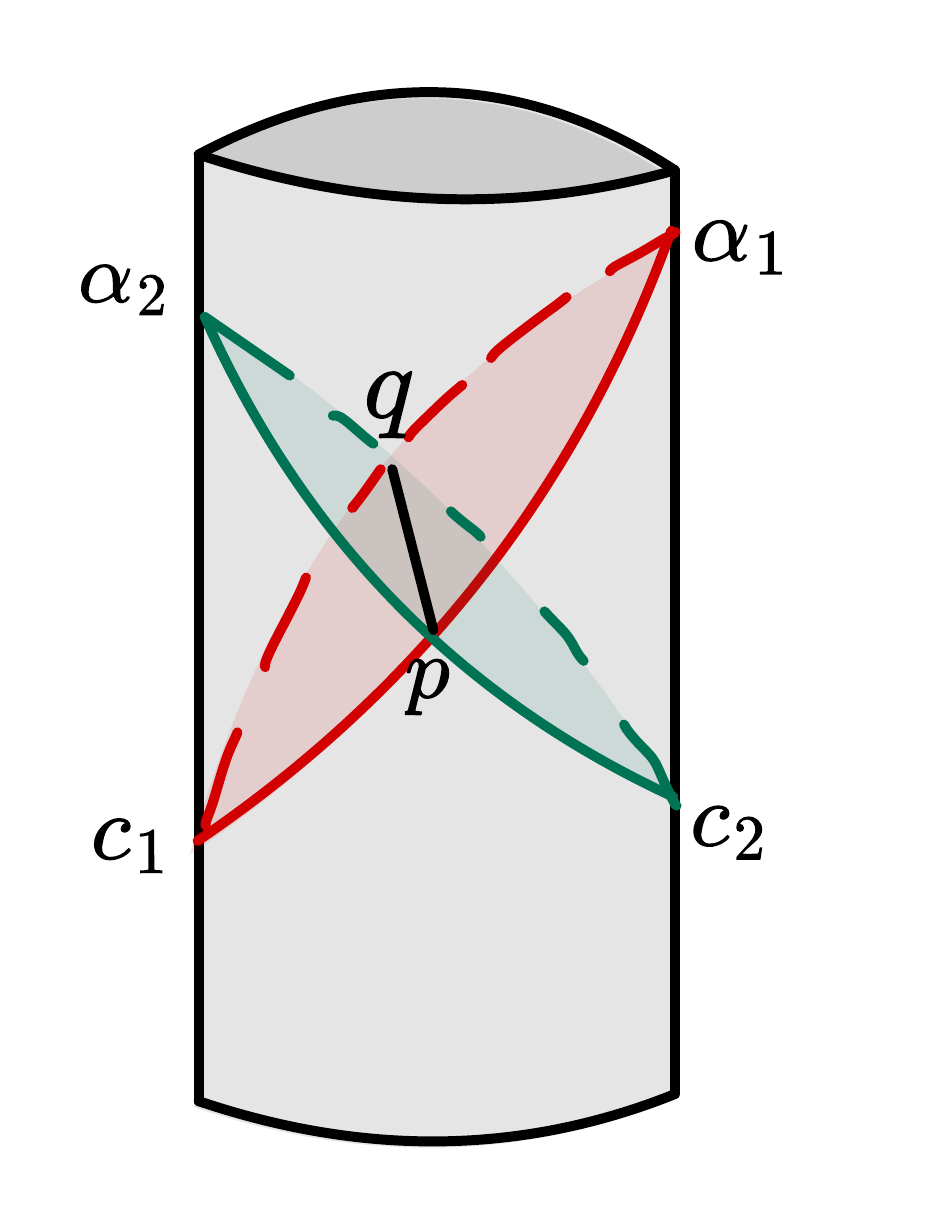}
    \caption{Illustration of two relevant null sheets intersecting at a simple ridge.}
    \label{fig:2Nridge}
\end{figure}

\begin{remark}
Lemma~\ref{lemma:2_null_ridge} obviously applies to
$\RR=\partial J^+[c_1]\cap \partial J^+[c_2]$.
In the following, we will apply the lemma to intersections of null sheets emanating from
HRRT surfaces that are anchored at $\partial \hat{J}^+[c_i]$, for example
$\RR_{V_1,V_2}=\partial J^+[RT(V_1)] \cap \partial J^+[RT(V_2)]$.

That two causal boundaries partition the full spacetime into four regions does not hold in
complete generality; counterexamples can be constructed in Minkowski space using compact
sets with non-convex boundaries.  The assumptions of the lemma exclude such cases.
\end{remark}

\begin{proof}
By assumption,
$\NN_i\cap \bb=\partial \hat{J}^+[c_i]$, so the two boundary intersection points
$p,q\in\hat{J}^+[c_1]\cap\hat{J}^+[c_2]$ belong to $\RR$. 

Let $\alpha_1$ be the future antipodal point of $c_1$, so that
$\partial \hat{J}^+[c_1]=\partial \hat{J}^-[\alpha_1]$ (see Figure \ref{fig:2Nridge} for an illustration).
This boundary null curve decomposes into two arcs:
one arc $\gamma_1$ connecting $p$ to $\alpha_1$ to $q$,
and the other arc $\gamma_2$ connecting $p$ to $c_1$ to $q$.
Both arcs lie on $\NN_1$.
By construction, $\gamma_1$ lies strictly to the future of $\NN_2$,
while $\gamma_2$ lies strictly to the past of $\NN_2$.

As a standard property of causal boundaries,
$\NN_1$ is ruled by null geodesic generators, and $\NN_2$ is achronal.
Consequently, each null generator of $\NN_1$ can intersect $\NN_2$ at most once:
if a generator intersected $\NN_2$ at two distinct points,
then the segment between them would lie entirely on one side of $\NN_2$,
allowing the construction of a timelike curve between two distinct points on $\NN_2$,
contradicting the achronality of $\NN_2$.

It follows that $\NN_2$ separates $\NN_1$ into exactly two connected components:
the portion to the future of $\NN_2$ and the portion to the past of $\NN_2$.
The arcs $\gamma_1$ and $\gamma_2$ lie in distinct components of
$\NN_1\setminus(\NN_1\cap\NN_2)$.

Therefore, the intersection $\RR=\NN_1\cap\NN_2$ must separate $\gamma_1$ from $\gamma_2$
within $\NN_1$, and hence must contain at least one connected component joining $p$ to $q$.
If $\RR$ contained more than one such component, or contained a closed-loop component in the
bulk, then some null generator of $\NN_1$ would necessarily intersect $\NN_2$ more than once,
contradicting the single-intersection property established above.

Thus $\RR$ consists of a single connected $p$--$q$ curve, which is necessarily spacelike,
simple, and entirely bulk-supported except at its endpoints.
\end{proof}

\begin{corollary}\label{lemma:3_ridge}
Let $c_1,c_2$ and $\beta$ be three distinct points on a boundary Cauchy slice of the timelike boundary $\bb$. Then the intersection of their boundary causal futures are pairwise nonemtpy.
    Consider three bulk causal boundaries $\NN_1$, $\NN_2$ and $\NN_3$ satisfying
    $$\NN_1\cap \bb=\hat{J}^+[c_1], \quad \NN_2\cap \bb=\hat{J}^+[c_2], \quad \NN_3\cap \bb=\hat{J}^+[\beta].$$
    Then the three null sheets intersect at a single point
    $$O=\RR_{\NN_1,\NN_2}\cap \RR_{\NN_1,\NN_3} = \RR_{\NN_1,\NN_2}\cap \RR_{\NN_2,\NN_3} =\RR_{\NN_1,\NN_3}\cap \RR_{\NN_2,\NN_3}  =\NN_1 \cap \NN_2 \cap \NN_3,$$
    where $\RR$ with subscripts denote the ridge of intersection between two relevant null sheets.
\end{corollary}

\begin{proof}
By Lemma~\ref{lemma:2_null_ridge}, each pair $\NN_i,\NN_j$ intersects along a unique ridge
\[
\RR_{\NN_i,\NN_j}:=\NN_i\cap \NN_j,
\]
which is a continuous, simple, spacelike curve with endpoints on $\bb$. In particular,
$\RR_{\NN_1,\NN_2}$ is nonempty.

\medskip
\noindent\textbf{Step 0: $\RR_{\NN_1,\NN_2}\cap \NN_3\neq\emptyset$.}
This follows from the boundary ordering: the endpoints of $\RR_{\NN_1,\NN_2}$ on $\bb$ lie in
different components of $\bb\setminus \hat J^+[\beta]$, so the connected curve
$\RR_{\NN_1,\NN_2}$ must cross $\NN_3$.

\medskip
\noindent\textbf{Step 1: $\RR_{\NN_1,\NN_2}\cap \NN_3$ cannot contain a curve segment.}
Suppose for contradiction that $\RR_{\NN_1,\NN_2}\cap \NN_3$ contains a nontrivial segment
$I$. Then $I\subset \NN_1\cap \NN_2\cap \NN_3$.
Assume there exists a point $x\in I$ at which all three hypersurfaces $\NN_1,\NN_2,\NN_3$ admit tangent planes and hence have well-defined null generator directions \footnote{Note that such a point $x$ necessarily exists: the ridge $\mathcal{R}_{\mathcal{N}_1,\mathcal{N}_2}$ is a spacelike curve whose points each lie on a unique null generator of $\mathcal{N}_1$ and $\mathcal{N}_2$; consequently no caustic forms along the ridge, and almost every point of $I$ has well‑defined tangent planes for all three null sheets.  We may therefore pick $x\in I$ where all three are smooth.}.

Because $I$ is spacelike (as a subcurve of the ridge $\RR_{\NN_1,\NN_2}$), its normal plane
$N_x I$ has Lorentzian signature $(-1,1)$ and therefore contains exactly two null directions.
For a differentiable null hypersurface $\NN_i$ containing $I$, its null generator at $x$
must be orthogonal to $T_x I$, hence must lie in $N_x I$ and coincide with one of these two
null directions. Therefore, among the three null sheets $\NN_1,\NN_2,\NN_3$ through $x$,
at least two have the same null generator direction at $x$.

Since those two hypersurfaces both contain the same spacelike segment $I$ through $x$ and
share the same null generator direction at $x$, they coincide in a neighborhood of $x$ as
ruled null hypersurfaces. This contradicts the distinct boundary anchoring data
$\NN_1\cap\bb=\hat J^+[c_1]$, $\NN_2\cap\bb=\hat J^+[c_2]$, $\NN_3\cap\bb=\hat J^+[\beta]$
for distinct boundary points. Hence $\RR_{\NN_1,\NN_2}\cap \NN_3$ contains no nontrivial segment.

\medskip
\noindent\textbf{Step 2: $\RR_{\NN_1,\NN_2}\cap \NN_3$ contains at most one point.}
By Step~1, $\RR_{\NN_1,\NN_2}\cap \NN_3$ is a discrete set of points.
Assume for contradiction that it contains two distinct points $O_1\neq O_2$.
Since \(\NN_3\) is an achronal causal boundary, it locally separates $\overline{M}$ into a future side and a past side.
If \(\RR_{\NN_1,\NN_2}\cap \NN_3\) contains two distinct points $O_1\neq O_2$, then along the connected spacelike curve \(\RR_{\NN_1,\NN_2}\) there exists an open segment whose interior lies entirely on one side of \(\NN_3\). This segment determines a “hole” between \(\RR_{\NN_1,\NN_2}\) and \(\NN_3\).

The ridge \(\RR_{\NN_1,\NN_2}\) is contained in both null hypersurfaces \(\NN_1\) and \(\NN_2\). At most one of these null hypersurfaces can contain the hole entirely. Since \(\NN_1\) and \(\NN_2\) are distinct causal boundaries, the other null hypersurface must intersect \(\NN_3\) on both sides of the hole, and hence must enter and leave one side of \(\NN_3\).

This contradicts Lemma~\ref{lemma:2_null_ridge} applied to the corresponding pair
$(\NN_1,\NN_3)$ or $(\NN_2,\NN_3)$, which guarantees that two such
causal boundaries separate the spacetime into four connected components.
Hence $\RR_{\NN_1,\NN_2}\cap \NN_3$ contains at most one point.

\medskip
\noindent\textbf{Step 3: Identification of the triple intersection point.}
Let $O$ denote the unique point in $\RR_{\NN_1,\NN_2}\cap \NN_3$.
Then $O\in \NN_1\cap\NN_2\cap\NN_3$, hence $O$ lies on each pairwise ridge. Therefore
\[
O=\RR_{\NN_1,\NN_2}\cap \RR_{\NN_1,\NN_3}
=\RR_{\NN_1,\NN_2}\cap \RR_{\NN_2,\NN_3}
=\RR_{\NN_1,\NN_3}\cap \RR_{\NN_2,\NN_3}
=\NN_1\cap\NN_2\cap\NN_3,
\]
as claimed.
\end{proof}

\begin{definition}[Geometric ordering of ridges]\label{def:ridge_order}
Let $c_1,c_2$ and $\beta_1,\beta_2$ be four distinct points on a boundary Cauchy slice of the
timelike boundary $\bb$.
Let $\NN_1=\partial J^+[\UU_1]$ and $\NN_2=\partial J^+[\UU_2]$ be two bulk future causal
boundaries satisfying
\[
\NN_1\cap \bb=\hat{J}^+[c_1], \qquad \NN_2\cap \bb=\hat{J}^+[c_2],
\]
and let $\NN_3=\partial J^-[\UU_3]$ and $\NN_4=\partial J^-[\UU_4]$ be two bulk past causal
boundaries satisfying
\[
\NN_3\cap \bb=\hat{J}^+[\beta_1], \qquad \NN_4\cap \bb=\hat{J}^+[\beta_2].
\]
Denote the ridges by
\[
\RR_{\UU_1,\UU_2}:=\NN_1\cap \NN_2, \qquad
\RR_{\UU_3,\UU_4}:=\NN_3\cap \NN_4 .
\]

We say that the ridge $\RR_{\UU_1,\UU_2}$ \emph{lies below} $\RR_{\UU_3,\UU_4}$ if
\[
J^+[\UU_1]\cap J^+[\UU_2]\cap J^-[\UU_3]\cap J^-[\UU_4]\neq \emptyset.
\]
Otherwise, we say that $\RR_{\UU_1,\UU_2}$ \emph{lies above} $\RR_{\UU_3,\UU_4}$, i.e.
\[
J^+[\UU_1]\cap J^+[\UU_2]\cap J^-[\UU_3]\cap J^-[\UU_4]= \emptyset.
\]
\end{definition}

\begin{remark}\label{rmk:ridge_order_geometry}
Definition~\ref{def:ridge_order} is a convenient shorthand for the causal relation between the
two ridges.  In the present framework, Lemma~\ref{lemma:2_null_ridge} implies that each ridge
is a unique simple curve with fixed boundary endpoints, and Corollary~\ref{lemma:3_ridge}
controls triple intersections of causal boundaries.
Together, these results exclude pathological configurations discussed in our previous
work~\cite{zhao2025proof}, such as:
\begin{itemize}
    \item the two ridges intertwining (e.g.\ one ``spiraling'' around the other);
    \item the two ridges intersecting at more than one point.
\end{itemize}
Indeed, any such configuration would force repeated entering/leaving behavior among causal
boundaries, contradicting Lemma~\ref{lemma:2_null_ridge}.
\end{remark}


\subsection{Characterization of connected entanglement wedges}\label{sec:characterize_connect}
Lastly, we recall some basic facts about a multipartite entanglement wedge being connected or multipartite mutual information being nonzero. We assume familiarity with these concepts as presented in Refs. \cite{witten2020introQI,wittenBHT}.

\begin{lemma}\label{lemma:connected_characterization}
Consider a union of disjoint subsets $V_1\cup \cdots \cup V_n$ and its causal complement $\Xall=(\Vall)^c$.
    The following are equivalent:
    \begin{enumerate}
        \item $\ew(V_1\cup \cdots \cup V_n)$ is connected.
        \item For any nontrivial (nonempty) bipartition of $V_1\cup \cdots \cup V_n=A\cup B$, the mutual information $I(A:B)>0$.
        \item $\ewX$ is fully disconnected \footnote{Since there are partially connected cases when $n>2$, we use fully disconnected to refer to the case of $\ewX=\ew(X_1)\cup \cdots \cup \ew(X_n)$.}.
    \end{enumerate}
\end{lemma}

    Further, one has 

\begin{lemma}\label{lemma:disconnect_lemma}
    Let $\ewX$ be fully disconnected, i.e.
    $\ewX=\ew(X_1)\cup \cdots \cup \ew(X_n)$, or in terms of entropy 
    $S (\Xall)=S(X_1)+ \cdots + S(X_n)$.
    Then any subset $A \subseteq \Xall$ also has fully disconnected entanglement wedge, i.e.
    $$\ew(A)= \cup_{X_i\in A} \ew(X_i),$$
    or in terms of entropy,
        $$S(A)= \sum_{X_i\in A} S(X_i).$$
\end{lemma}
\begin{proof}
    If $\ew(A)$ is connected or partially connected, then its HRRT surfaces are composed of a union of surfaces with strictly smaller total area than $\cup_{X_i\in A} RT(X_i)$. Combining with $\cup_{X_j\notin A} RT(X_j)$, this gives a candidate HRRT surface for $\Xall = A\cup A^c$ with strictly smaller area than $\cup_i RT(X_i)$. This would contradict the condition that $\ewX$ is fully disconnected.
\end{proof}

\begin{lemma}\label{lemma:connect_lemma}
    Consider a union of disjoint subsets $V_1\cup \cdots V_n$. If the entanglement wedge of any pair is connected, or equivalently, $$I(V_i:V_j)>0, \forall i\neq j$$ 
    then, $\ewV$ is connected.
\end{lemma}
\begin{proof}
By Lemma~\ref{lemma:connected_characterization}, to show that $\ewV$ is connected it suffices to show that $I(A:B)>0$ for every nontrivial bipartition $V_1\cup\cdots\cup V_n=A\cup B$. Then the Lemma follows directly from the monotonicity of mutual information or strong subadditivity: for any $V_i\subseteq A, V_j\subseteq B$ one has $I(A:B)\geq I(V_i:V_j)>0$.
\end{proof}
\begin{remark}
If one only assumes that $\ew(V_1\cup \cdots \cup V_n)$ is connected, the entanglement wedge of any pair could be disconnected. Simple examples can be constructed in pure $AdS_3$.
\end{remark}

In light of the $\GMay$ graph introduced by ref. \cite{may2022nton}, we can introduce a mutual information graph $\GMI$ whose vertices $1,2, \cdots n$ represent the input regions $V_i$'s and whose edge $i-j$ represents $I(V_i:V_j)\sim O(1/G_N)$. The following lemma could make the interpretation of Theorem \ref{thm:n-n-CWT} of ref. \cite{may2022nton} more transparent.
\begin{lemma}\label{lemma:MI_EW}
    If the mutual information graph $\GMI$ is connected,
    then $\ewV$ is connected. 

    A special case is that pairwise mutual information $I(V_i:V_j)>0, \forall i\neq j$.
\end{lemma}
\begin{proof}
    For a set $V_1\cup \cdots \cup V_n$ to have connected entanglement wedges, any nontrivial (nonempty) bipartition of $V_1\cup \cdots \cup V_n=A\cup B$ should have strictly positive mutual information, i.e. 
    $$I(A:B)>0.$$
    
    If $\GMI$ is connected, then for any bipartition $\Vall=A\cup B$, there would exist $V_i\in A$ and $V_j \in B$ such that $I(V_i:V_j)>0$.
    Then applying inductively the monotonicity of mutual information $I(V_1: V_2\cup V_3)\geq I (V_1:V_2)$, one would have
    $$I(A:B)\geq I(V_i:V_j)>0.$$
\end{proof}

In fact, as shown in Ref. \cite{hernandezHubeny2024NMI}, one can give a complete information theoretical characterization of connected multipartite entanglement wedges using multipartite mutual information
\begin{equation}
    I_n(V_1:\cdots:V_n)= \sum_{k=1}^n (-1)^{k-1} \sum_{i_1<\cdots<i_k}S(V_{i_1} \cup \cdots \cup V_{i_k}).
\end{equation}
For example, the monogamy of mutual information (MMI)
\begin{equation}
    S(V_1)+S(V_2)+S(V_3)+S(V_1\cup V_2\cup V_3)\geq S(V_1\cup V_2) +  S(V_1\cup V_3) +  S(V_2\cup V_3)
\end{equation}
which holds for holographic states but not general quantum states can be recast as $$-I_3(V_1\cup V_2\cup V_3) \geq 0.$$
We document here a theorem of ref. \cite{hernandezHubeny2024NMI} for completeness.
\begin{theorem}\label{thm:NMI-wedge}
   The multipartite information of $n$ disjoint subsystems $V_i$ in a generic configuration \footnote{By generic configuration, they exclude phase transition cases when multiple extremal surfaces exchange dominance.} is nonvanishing $I_n(V_1:\cdots:V_n)\neq 0$ if and only if the joint entanglement wedge of $\ewV$ is connected.    
\end{theorem}

\section{Generalizing the n-to-n Connected Wedge Theorem}\label{sec:generalize_n}

\subsection{An improvement of the n-to-n Connected Wedge Theorem}\label{sec:may2022}
Let us start with giving a slightly different proof of the $n$-to-$n$ Connected Wedge Theorem, which also reveals that the necessary condition of Theorem \ref{thm:n-n-CWT} can be further weakened, as stated in Theorem \ref{thm:weaker-necessary}.

\begin{figure}
    \centering
    \includegraphics[width=0.8\linewidth,trim={5 9cm 5 5},clip]{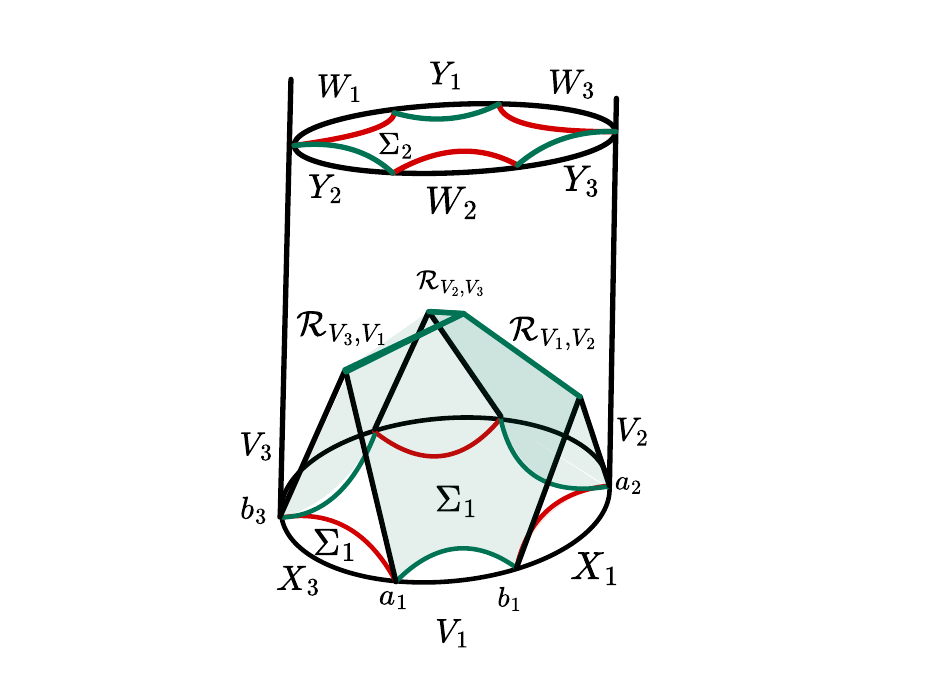}
    \caption{Illustration of the surface $\ZZi$ formed by future null sheets emanating from $RT(V_i)$ for $n=3$. Green curves in $\Sigma_1$ and $\Sigma_2$ lie on $\partial J^+[\ew(V_i)]$ while red curves in $\Sigma_1$ and $\Sigma_2$ lie on $\partial J^-[\ew(W_j)]$. The ridges $\RR_{V_i,V_{i+1}}=\partial J^+[{\ew(V_i)}]\cap \partial J^+[{\ew(V_{i+1})}]$ are labeled explicitly. Note that $n+1=1$ because of the $S^1$ topology.}
    \label{fig:ZZin}
\end{figure}

    For the purpose of explanation, let us pick $\Sigma_1$ to be a bulk Cauchy slice bounded by $\hat{\Sigma}_1$ and containing relevant HRRT surfaces\footnote{This is always possible because HRRT surfaces of disjoint spacelike-separated boundary regions can be minimal on the same Cauchy slice \cite{wall2014maximin}.}. Similarly, let $\Sigma_2$ be a bulk Cauchy slice bounded by $\hat{\Sigma}_2$ that contains relevant HRRT surfaces.
    
We consider the geometric surface $\ZZi:=\partial J^+[\ew(V_1\cup\cdots\cup V_n)]$. In the fully disconnected case, the surface $\ZZi$ is formed by the union of all future-pointing null sheets $\NN_{V_i} = \partial J^+[RT(V_i)]$ that emanate from $RT(V_i)$ \footnote{We note that by $\partial J^+[RT(V_i)]$ etc., we only refer to the future null sheet emanating from HRRT surfaces that points toward the bulk.}, truncated at their mutual intersections (see Figure \ref{fig:ZZin} for an illustration with $n=3$). The surface $\ZZi$ is therefore made up of null sheets intersecting at a net of vertices and (subsets of) ridges.
Noting that $\NN_{V_i}$ is also the future horizon of $\ew(V_i^c)$\footnote{Recall that superscript $c$ indicates causal complements.}, this surface $\ZZi$ is also the future boundary of the compact set $\cap_{i=1}^{n} \, \ew(V_i^c)$.
Also note that $\ZZi$ and $\Sigma_1$ bound a compact subset of $\overline{M}=M \cup \bb$, which will be denoted by $\ZZ_B$.

We also consider past-pointing null sheets $\NN_{W_j} = \partial J^-[RT(W_j)]$ that emanate from $RT(W_j)$'s. Note that $\NN_{W_j}$ is also the past horizon of $\ew(W_j^c)$. By the causal anchoring principle, $\NN_{V_i}\cap \bb$ and $\NN_{W_i}\cap \bb$, are just light rays emitted from $c_i$ and $r_j$ (or $\beta_j$), respectively.

Due to cyclic boundary ordering $V_1, X_1, \cdots V_n, X_n$ (and similarly for the outputs), the null sheet $\NN_{W_i}$, which is anchored at the spatial boundaries $b_i$ and $a_{i+1}$ of $X_i$, is topologically constrained to intersect the surface $\ZZi$. We assert that this intersection $\CC_i$, is a simple curve on $\ZZi$ with endpoints precisely at $b_i$ and $a_{i+1}$.
The argument proceeds in two steps. First, within the bulk region $\ZZ_B$ bounded by $\ZZi$ and the Cauchy slice $\Sigma_1$, the intersection $\CC_i$ is homotopic to $\NN_{W_i} \cap \Sigma_1$, which is itself homotopic to the HRRT surface $RT(X_i)$ within $\Sigma_1$. Second, Lemma \ref{lemma:2_null_ridge} excludes the possibility that $\CC_i$ contain homotopically trivial closed loops: $\NN_{W_j}\cap \NN_{V_i}$ is a simple ridge in $\NN_{V_i}$ and $\ZZi$ is made up of intersecting null sheets $\NN_{V_i}$'s. This concludes that $\CC_i$ is a simple curve without disconnected components. A simple illustration for the $n=3$ case is provided in Figure \ref{fig:focusing_may2022}(a). Moreover, each curve $\CC_i$ must have segments in at least the adjacent null sheets $\NN_{V_i}$ and $\NN_{V_{i+1}}$, though it may also cross other sheets that constitute $\ZZi$. 

\begin{figure}
    \centering
    \includegraphics[width=0.8\linewidth]{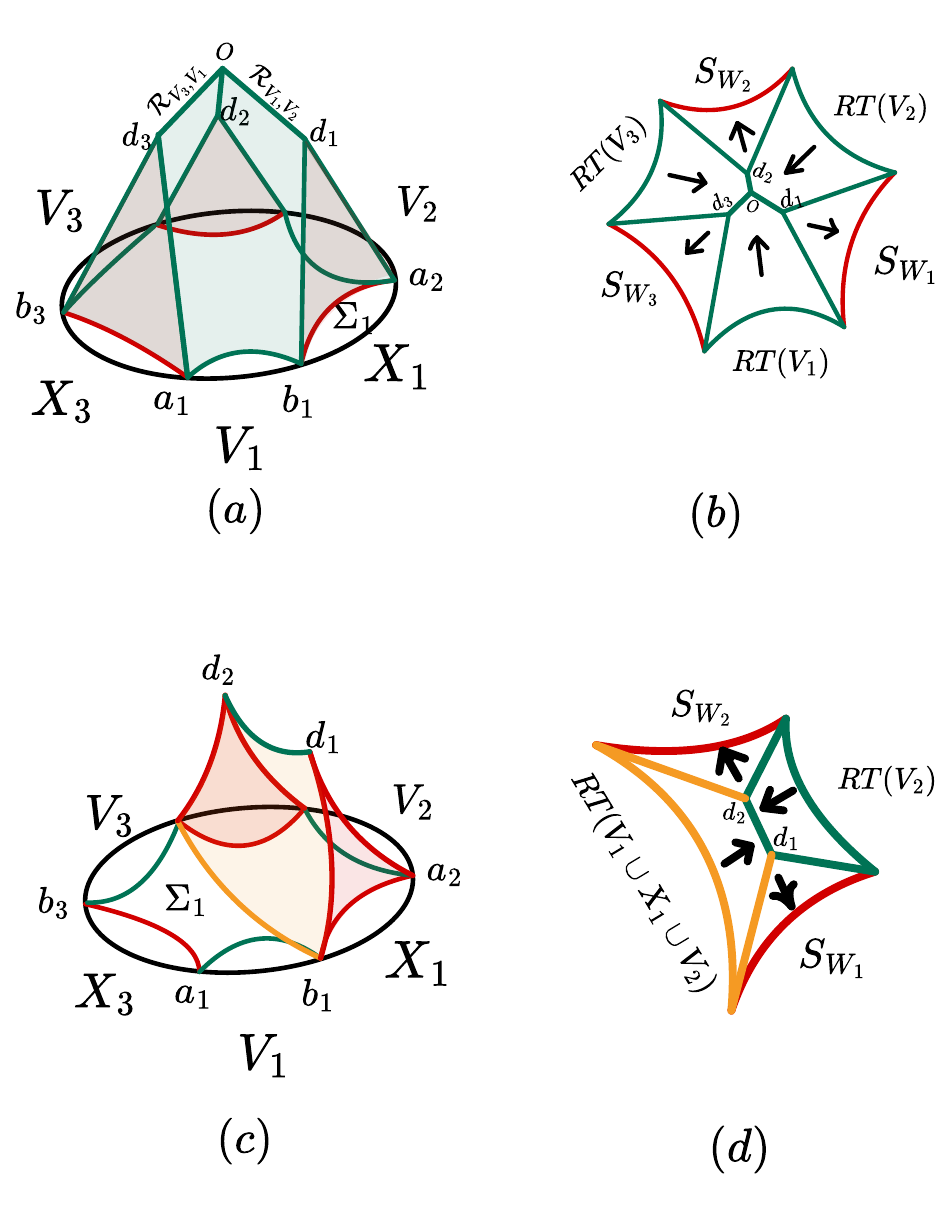}
    \caption{Illustration of focusing calculation when $n=3$. Panel (a) shows the geometric structure $\ZZi$ cut by $\partial J^-[\ew(W_i)]$ (compare with Figure \ref{fig:ZZin}), where future null sheets emanating from $RT(V_i)$ are shown in green and past null sheets emanating from $RT(W_i)$ are shown in red. The curves $\CC_i=\partial J^-[\ew(W_i)]\cap \ZZi$ is shown as $b_i-d_i-a_{i+1}$ (counting modulo $n$). Panel $(b)$ show the $\ZZi$ in a flattened fashion where $S_{W_i}$ is the intersection of null sheet $\partial J^-[\ew(W_i)]$ with $\Sigma_1$, i.e. $S_{W_i}=\partial J^-[\ew(W_i)]\cap \Sigma_1$. Arrows indicate the direction along which null expansion $\theta$ decreases. Panels (c) and (d) are similar to panels $(a)$ and $(b)$ but for partially connected scenarios.}
    \label{fig:focusing_may2022}
\end{figure}

We now prove that if all curves $\CC_i$ are distinct, the entanglement wedge  cannot cannot decompose as a disjoint union $\ew(V_1)\cup \cdots\cup \ew(V_n)$. The proof constructs a specific geometric comparison on the null surface $\ZZi$, following a strategy analogous to the original CWT proof \cite{maypenington2020holographic}.

The key observation is that on $\ZZi$, the collection $\bigcup_{i=1}^n \CC_i$ is homologous to $\bigcup_{i=1}^n \mathrm{RT}(V_i)$.  
Indeed, each null‑sheet component $\mathcal{N}_{V_i}$ of $\ZZi$ is bounded by $\mathrm{RT}(V_i)$ in the past and by segments of the ridges $\mathcal{R}_{V_i,V_{i\pm1}}$ (and possibly others) in the future.  
The curves $\CC_{i-1}$ and $\CC_i$ intersect $\mathcal{N}_{V_i}$ along simple arcs that connect the boudnary $a_i$ or $b_i$ with ridge points; because $\ZZi$ is a topological disk, these arcs divide the boundary of $\mathcal{N}_{V_i}$ into two regions --- one containing $\mathrm{RT}(V_i)$, the other containing the remaining ridge network.
Hence $\mathrm{RT}(V_i)$ is homologous within $\mathcal{N}_{V_i}$ to a composite curve $\tilde\gamma_i$ built from the portions of $\CC_{i-1}$ and $\CC_i$ on this sheet together with the ridge segments that join their endpoints.
Summing over $i$ yields the desired homology on $\ZZi$, with the interior ridge segments appearing twice and cancelling in homology.
For $n=3$ (Figure \ref{fig:focusing_may2022}), $\tilde\gamma_1$ would be the union of $a_1-d_3$ (from $\CC_3$), $b_1-d_1$ (from $\CC_1$), and the ridge segments $O-d_3$ and $O-d_1$.

An area comparison follows from the focusing property ($\theta\leq 0$) on the null sheets emanating from extremal surfaces.
Consider a single null sheet $\NN_{V_i}$.
The composite curve $\tilde\gamma_i$ is made up of the following segments:
\begin{itemize}
    \item the portion of $\CC_{i-1}$ that lies on $\NN_{V_i}$;
    \item the portion of $\CC_i$ that lies on $\NN_{V_i}$;
    \item ridge segments (if any) connecting the exit points of $\CC_{i-1}$ 
          and $\CC_i$ on the ridges that bound $\NN_{V_i}$.
\end{itemize}
Moving along the future‑directed null generators of $\NN_{V_i}$ away from the HRRT surface $\mathrm{RT}(V_i)$,  the non‑positive expansion $\theta\leq0$ implies that areas of homologous sections decrease. Therefore, we obtain
\begin{equation}\label{eq:focusing_may_V}
    |\tilde\gamma_i| \;\leq\; |\mathrm{RT}(V_i)| .
\end{equation}
Moreover, by the explicit construction of the $\tilde\gamma_i$, 
the union of all original intersection curves is contained in the union 
of all $\tilde\gamma_i$:
\[
    \bigcup_{i=1}^{n} \CC_i \;\subseteq\; \bigcup_{i=1}^{n} \tilde\gamma_i ,
\]
with strict inclusion whenever a ridge segment appears in some $\tilde\gamma_i$.

On the past‑pointing null sheets $\NN_{W_i}$, we push each $\CC_i$ backwards 
along the generators until it reaches the Cauchy slice $\Sigma_1$.
Again by $\theta\leq0$ along the past direction, and using that 
$\mathrm{RT}(X_i)$ is the minimal‑area surface on $\Sigma_1$ homologous 
to $X_i$, we have
\begin{equation}\label{eq:focusing_may_W}
    |\CC_i| \;\geq\; |\NN_{W_i}\cap\Sigma_1| \;\geq\; |\mathrm{RT}(X_i)| .
\end{equation}

Combining the two estimates and summing over $i$ gives
\begin{equation}\label{eq:focusing_may_combine}
    \sum_i |\mathrm{RT}(V_i)|
    \;\geq\; \sum_i |\tilde\gamma_i|
    \;\geq\; \sum_i |\CC_i|
    \;\geq\; \sum_i |\mathrm{RT}(X_i)| ,
\end{equation}
where the first inequality is from \eqref{eq:focusing_may_V} and the inclusion 
of the $\CC_i$, and the last two follow from \eqref{eq:focusing_may_W}.
If any ridge segment is present in the $\tilde\gamma_i$, the inclusion is 
strict and the inequality becomes strict, contradicting the hypothesis that 
$\ewV$ is connected (which would require the opposite strict inequality).
Therefore $\ewV$ cannot be fully disconnected.

We next address partially connected configurations. Recall that \(Z_{\mathrm{in}}=\partial J^+[\ew(V_1\cup\cdots\cup V_n)]\) by definition; in partially connected phases this surface is generated by null sheets from HRRT surfaces of enlarged input regions.
These enlarged inpute regions take the form $V_k\cup X_k\cup \cdots V_{k+l}$ and remain separated by a subset of the original complement regions $X_j$'s. A key consequence of entanglement wedge nesting is that this new surface $\ZZi$ lies inside the compact region $\ZZ_B$. In other words, this new $\ZZi$ lies nowhere to the future of the original $\ZZi$ formed solely from $\NN_{V_i}$.

This inclusion relation ensures that the new intersection curves $\CC_j=\NN_{W_j}\cap \ZZi$ (defined for the set of $j$ with $X_j$ remaining as complements of the enlarged input regions) remain simple and distinct from one another. With this geometric structure in place, one can perform a length/area comparison analogous to the previous connected case. This calculation leads to the conclusion that the length of HRRT surfaces for the enlarged input regions must exceed the sum $\sum |RT(X_j)|$, where the sum runs only over the $X_j$ appearing as complements. This inequality presents a contradiction, thereby proving that $\ewV$ cannot be in a partially connected state.
  
The preceding geometric proof establishes a general criterion: any condition that ensures the intersection curves $\CC_i$ are distinct on $\ZZi$ is sufficient to conclude that $\ewV$ is connected. 

Two curves $\CC_i$ and $\CC_j$ intersect exactly where the ridge
$\RR_{W_i,W_j} = \partial J^-[RT(W_i)] \cap \partial J^-[RT(W_j)] $
enters $\ZZi$. Hence
\[
\CC_i \cap \CC_j \neq \emptyset \quad\Longleftrightarrow\quad
\RR_{W_i,W_j} \cap \ZZi \neq \emptyset .
\] 
Thus ``all $\CC_i$ are distinct'' is equivalent to the condition that every output ridge lies strictly above the surface $\ZZi$:
\begin{equation}\label{eq:weaker_necessary_1}
\forall\, i\neq j,\qquad \RR_{W_i,W_j} \cap \ZZi = \emptyset \quad \text{or} \quad \big(\bigcap_{k=1}^n \, \ew(V_k^c)\big) \, \cap \,\big(\ew(W_i^c)\cap \ew(W_j^c)\big)=\emptyset.
\end{equation}
In geometric ordering language, this means that none of the output ridges $\RR_{W_k,W_l}$ enters the region enclosed by the future horizon
$\ZZi = \partial J^+\bigl[\bigcap_i \ew(V_i^c)\bigr]$.

A particularly simple condition that guarantees \eqref{eq:weaker_necessary_1} is the single‑pair $2$-to‑all criterion:
\begin{equation}\label{eq:weaker_necessary_main}
    \exists k\neq l \, \text{ such that } J^+[\ew(V_k)]\cap J^+[\ew(V_l)]\cap  \, \bigcap_{i=1}^{n} J^-[\ew(W_i)] \neq \emptyset
\end{equation}
i.e., the ridge $\RR_{V_k,V_l}$ lies below all output ridges
$\cup_{i\neq j}\RR_{W_i,W_j}$.
In contrast to the connected $2$-to‑all graph of
Ref.~\cite{may2022nton}, which demands that at least $n-1$ such pairs exist to form a connected graph, condition~\eqref{eq:weaker_necessary_main} involves only a
single pair and is therefore strictly weaker. 
%


\medskip
\noindent
\textbf{Time‑reversed sufficient condition for the output wedge.}
The entire geometric setup is time‑reversal symmetric: exchanging $c_i \leftrightarrow r_i$, $V_i \leftrightarrow W_i$, $X_i \leftrightarrow Y_i$ and future $\leftrightarrow$ past maps the $n$‑to‑$n$ problem onto itself.
Consequently, the argument that showed $\mathcal{E}(V_1\cup\cdots\cup V_n)$ is connected under condition~\eqref{eq:weaker_necessary_main} also yields a
completely analogous sufficient condition for the connectedness of $\mathcal{E}(W_1\cup\cdots\cup W_n)$.
In detail:

\begin{itemize}
    \item The analogue of the ``output ridges above $\mathcal{Z}_{\rm in}$''
    criterion~\eqref{eq:weaker_necessary_1} becomes: all \emph{input} ridges
    $\mathcal{R}_{V_i,V_j}$ lie strictly below the past horizon
    $\mathcal{Z}_{\rm out} = \partial J^{-}\bigl[\bigcap_k \mathcal{E}(W_k^c)\bigr]$,
    i.e.
    \[
    \forall\, i\neq j,\qquad
    \mathcal{R}_{V_i,V_j} \cap \mathcal{Z}_{\rm out} = \emptyset
    \quad\text{or}\quad
    \bigl(\bigcap_{k=1}^n \mathcal{E}(W_k^c)\bigr)
    \cap \bigl(\mathcal{E}(V_i^c)\cap \mathcal{E}(V_j^c)\bigr) = \emptyset .
    \]

    \item The analogue of the explicit $2$‑to‑all condition~\eqref{eq:weaker_necessary_main}
    is the existence of a single pair $k\neq l$ such that
    \[
    \bigcap_{i=1}^{n} J^{+}[\mathcal{E}(V_i)]
    \;\cap\;
    J^{-}[\mathcal{E}(W_k)]\cap J^{-}[\mathcal{E}(W_l)]
    \;\neq\; \emptyset .
    \]
    In words, the ridge $\mathcal{R}_{W_k,W_l}$ lies above all input ridges
    $\cup_{i\neq j}\mathcal{R}_{V_i,V_j}$.
\end{itemize}

\noindent
If either of these time‑reversed conditions holds, the output entanglement
wedge $\mathcal{E}(W_1\cup\cdots\cup W_n)$ is connected.

\subsubsection{Interpretation of our weaker condition}
Let us compare Theorem \ref{thm:weaker-necessary} or Equation \eqref{eq:weaker_necessary_main} with the framework of ref.~\cite{may2022nton}. In that setting, a causal relation of the form
\[V_i,V_j \to \{W_k\}_{k=1}^n\]
is understood as a primitive bulk channel: there exists a spacetime region that can receive qantum information from the input wedges $\ew(V_i)$ and $\ew(V_j)$, and subsequently redistribute it to all output wedges $\ew(W_k)$. Within the semiclassical code subspace, such primitives can be composed into a causal network that allows the implementation of general quantum information tasks in the bulk \footnote{This statement should be understood at the level of causal structure, without imposing resource constrains. In a fixed semiclassical geometry, one should restrict to protocols that can be implemented without generating an extensive amount of entanglement that would significantly backreact on the geometry. For this reason, ref. \cite{may2022nton} focuses on B84-type task, which require only $O(n)$ entangled pairs.}.

From this perspective, the arguments of ref. \cite{may2022nton} suggest that multipartite entanglement on the boudnary arises when a sufficiently rich network of such $2$-to-all channels is available. Theorem \ref{thm:weaker-necessary} shows that, in holography, this expectation is in fact overly conservative. A single $2$-to-all causal relation, involving only a pair of input regions, is already sufficient to force the entire input entanglement wedge $\ew(\Vall)$ to be connected.

From the boundary perspective, this implies that every nontrivial bipartition of the input regions carries mutual information of order \(O(1/G_N)\).
Thus, although the bulk causal condition only involves a single pair \(V_i,V_j\), its consistent reliazation in the boundary theory requires a genuinely multipartite entanglement resource, extending across all input regions.

In this sense, the theorem identifies a rigidity phenomenon: once one $2$-to-all bulk causal channel penetrates sufficiently deep through the entanglement wedges, the focusing of null congruences and the homology constraint on HRRT surfaces prevent the remaining input regions from remaining geometrically disconnected.

\subsection{Consequences of connected entanglement wedges}\label{sec:connect_cons}
We now discuss consequences of $\ewV$ being connected.
When $\ewV$ is connected, it is natural to work with the future horizon of $\ewV$, which is formed by future null sheets $\NN_{X_i}$ emanating from $RT(X_i)$.
Accordingly, we compare these sheets with past-pointing null sheets emanating from $RT(Y_i)$, and we formulate the consequences below in terms of \emph{ridges} entering $\ewV$.
\subsubsection{Consequence from a connected--disconnected area comparison}
\label{sec:cons-conn-disc}
We first derive a necessary geometric consequence of $\ewV$ being connected by a contradiction
argument comparing the connected and fully disconnected phases:
\begin{equation} \label{eq:connected_cons1_1}
    \exists\, k\neq l \ \text{such that}\ 
    \ew(Y_k^c)\cap \ew(Y_l^c)\cap \ewV \neq \emptyset .
\end{equation}
Equivalently, there exists a ridge
\[
\RR_{Y_k,Y_l}:=\partial J^-[RT(Y_k)]\cap \partial J^-[RT(Y_l)]
\]
that \emph{enters} $\ewV$, or, in the ridge-ordering language,
lies below all ridges $\RR_{X_i,X_j}:=\partial J^+[RT(X_i)]\cap \partial J^+[RT(X_j)]$
for all $i\neq j$ (since $\ewV=\bigcap_{i=1}^{n} \ew(X_i^c)$).

\paragraph{(Optional) scattering interpretation.}
If we additionally assume $\ewW$ is connected, then by Lemma~\ref{lemma:connected_characterization}
this is equivalent to $\ewY$ being fully disconnected. In that case, for any pair $k\neq l$,
$\ew(Y_k\cup Y_l)$ is disconnected (Lemma~\ref{lemma:disconnect_lemma}), i.e.\ the two components
of $(Y_k\cup Y_l)^c$ have connected entanglement wedges. We refer to these two components as
\emph{modified output regions}:
\begin{align}\label{eq:modified_output}
    \tilde{W}_k &= W_k \cup Y_{k+1}\cup \cdots \cup W_{l-1}, \nonumber\\
    \tilde{W}_l &= W_l \cup Y_{l+1}\cup \cdots \cup W_{k-1},
\end{align}
with indices understood modulo $n$.
We define the corresponding \emph{modified output points} $\tilde r_k,\tilde r_l$ by
\begin{align}
    \tilde{W}_k
    &= \hat{J}^-(\tilde r_k)\cap\hat{J}^+(c_k)\cap\hat{J}^+(c_l),
    \label{eq:tilde-rk-def1}\\
    \tilde{W}_l
    &= \hat{J}^-(\tilde r_l)\cap\hat{J}^+(c_k)\cap\hat{J}^+(c_l),
    \label{eq:tilde-rl-def1}
\end{align}
where $c_k,c_l$ are the corresponding input points.
See Figure~\ref{fig:reduced_3_2}(a) for an illustration when $n=3$.

With these definitions, \eqref{eq:connected_cons1_1} can be rephrased as
\begin{equation}\label{eq:connected_cons1_2}
    \exists\, k\neq l \ \text{such that}\ 
    \ewV\cap \ew(\tilde{W}_k \cup \tilde{W}_l)\neq \emptyset .
\end{equation}
Or in pairwise terms,
\begin{equation}\label{eq:connected_cons1_3}
    \exists\, k\neq l \ \text{such that}\ 
    \forall\, i\neq j,\quad
    \ew(\tilde{V}_i\cup \tilde{V}_j)\cap \ew(\tilde{W}_k \cup \tilde{W}_l)\neq \emptyset ,
\end{equation}
where $\tilde{V}_i\cup \tilde{V}_j=(X_i\cup X_j)^c$.

\paragraph{Proof of \eqref{eq:connected_cons1_1}.}
We argue by contradiction.
Assume that for all $k\neq l$,
\[
\ew(Y_k^c)\cap \ew(Y_l^c)\cap \ewV=\emptyset,
\]
equivalently every ridge $\RR_{Y_k,Y_l}$ lies strictly above $\ewV$.
Let $\ZZi$ denote the future horizon of $\ewV$, i.e.\ the surface formed by future null
sheets emanating from the $RT(X_i)$ up to their intersections. As in Section~\ref{sec:may2022},
each past horizon $\NN_{Y_k}:=\partial J^-[\ew(Y_k)]$ intersects $\ZZi$ in a simple curve
$\CC_k$ with endpoints $a_k,b_k$. If all ridges lie above $\ZZi$, then these curves are all
distinct: intersections $\CC_k\cap\CC_l$ would correspond to where the ridge
$\NN_{Y_k}\cap \NN_{Y_l}$ enters/leaves $\ewV$\footnote{
One can show any ridge enters and leaves $\ZZi$ at most once using Lemma~\ref{lemma:2_null_ridge}
and Corollary~\ref{lemma:3_ridge}.}.
In this situation one repeats the focusing calculation of Section~\ref{sec:may2022} to obtain
\begin{equation}\label{eq:focusing_BZ_disconnect}
    |RT(X_i)| \ge |RT(V_i)|,
\end{equation}
contradicting that $\ewV$ is connected. This establishes \eqref{eq:connected_cons1_1}.
See Figure~\ref{fig:CWT_one_ridge}(a--b) for an illustration when $n=3$.

\begin{figure}
    \centering
    \includegraphics[width=0.8\linewidth]{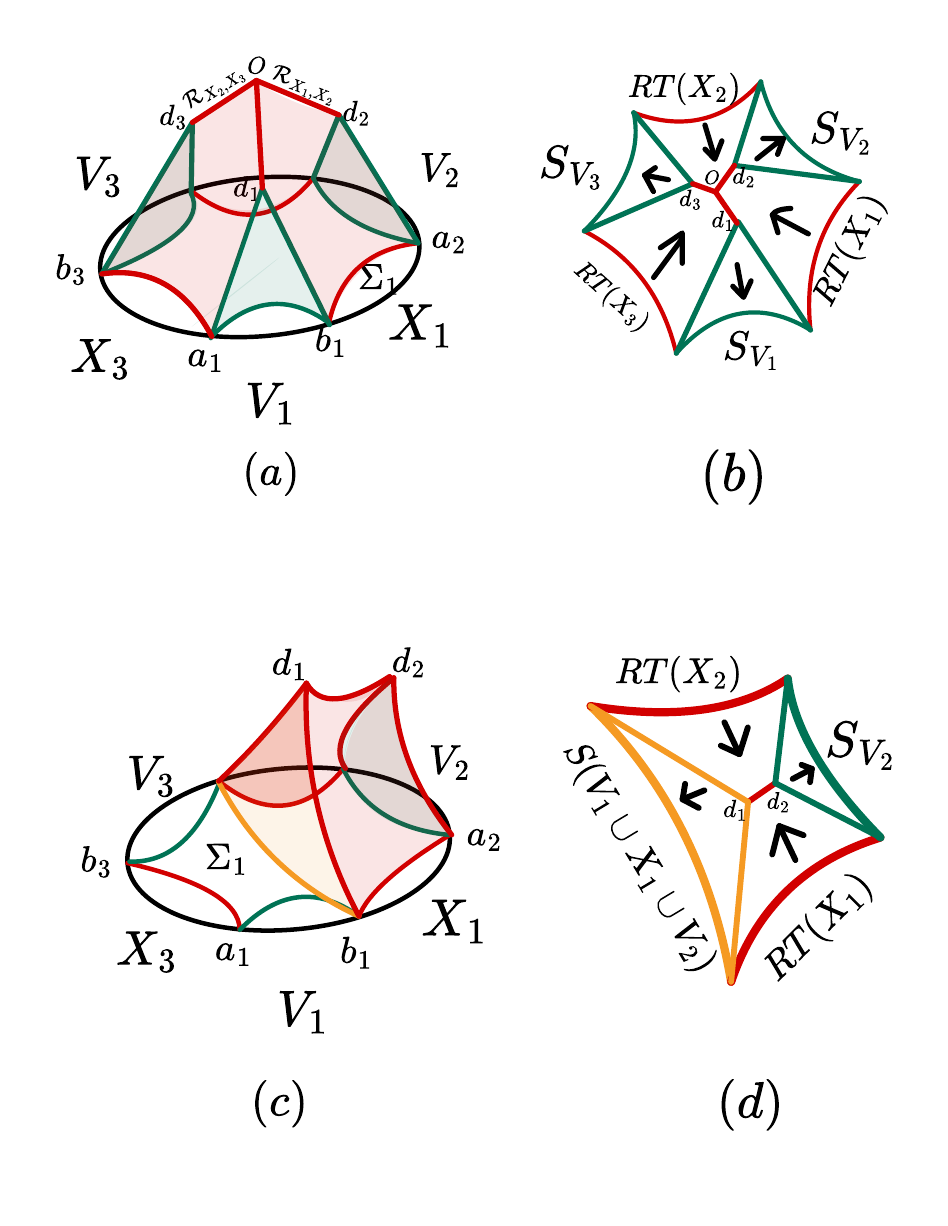}
    \caption{Illustration of the geometric structure used in deriving consequences of connected entanglement wedge $\ewV$. Panel (a) shows the geometric structure $\ZZi$ (future horizon of $\ewV$) cut by $\partial J^-[\ew(Y_i)]$, where future null sheets emanating from $RT(X_i)$ are shown in red and past null sheets emanating from $RT(Y_j)$ are shown in green. The curves $\CC_i=\partial J^-[\ew(Y_i)]\cap \ZZi$ is shown as $a_i-d_i-b_i$. Panel $(b)$ show the $\ZZi$ in a flattened fashion where $S_{V_i}$ is the intersection of null sheet $\partial J^-[\ew(Y_i)]$ with $\Sigma_1$, i.e. $S_{V_i}=\partial J^-[\ew(Y_i)]\cap \Sigma_1$. Arrows indicate the direction along which null expansion $\theta$ decreases. Panels $(c)$ and $(d)$ are similar to panels $(a)$ and $(b)$ but for enlarged input regions.}
    \label{fig:CWT_one_ridge}
\end{figure}

\begin{figure}
    \centering
    \includegraphics[width=0.8\linewidth]{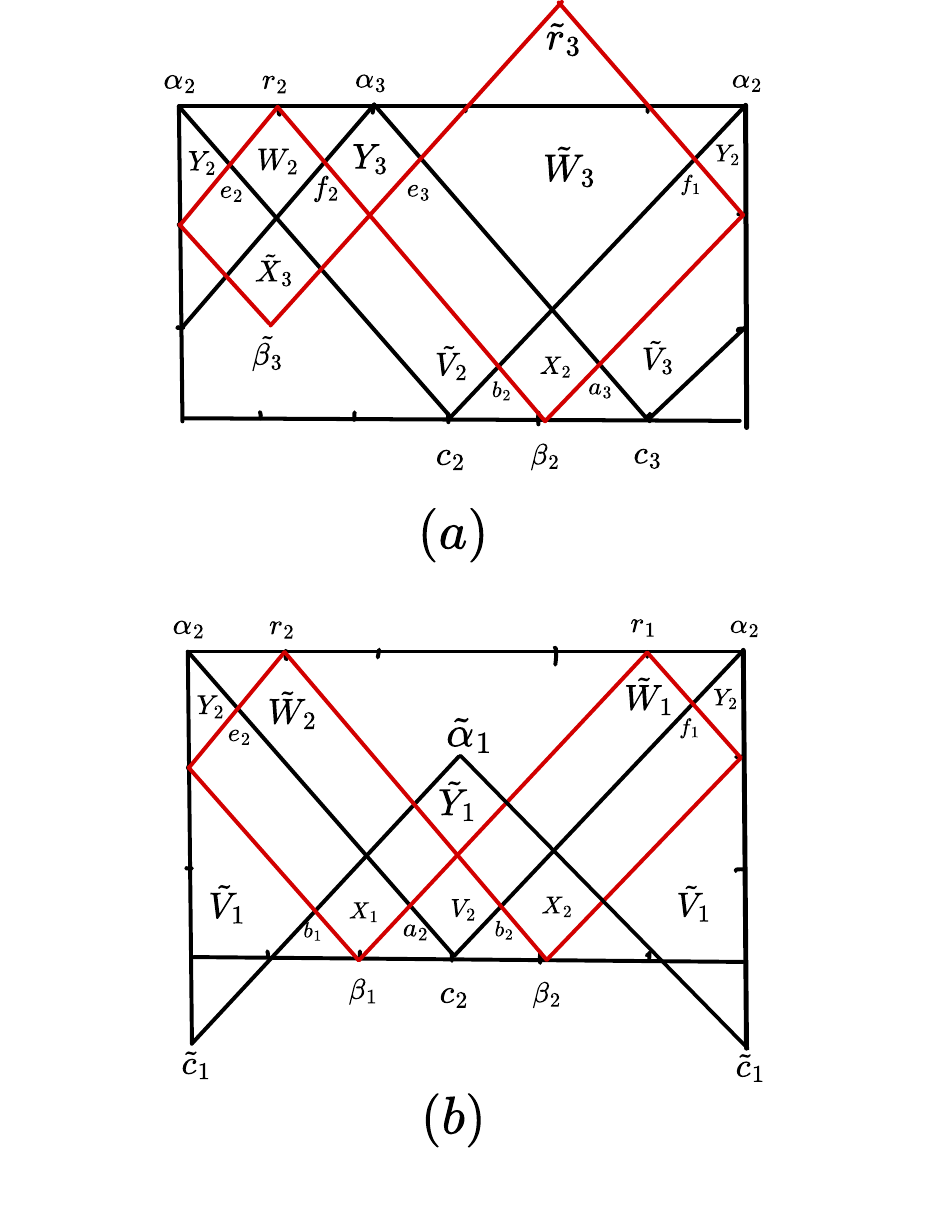}
    \caption{Illustration of modified input and output regions for $n=3$. Panel (a) show that $RT(Y_2)\cup RT(Y_3)$ can be regarded as HRRT surfaces of $\ew(W_2\cup \tilde{W}_3)$, where $\tilde{W}_3 \supseteq W_3\cup Y_1\cup W_1$. We get an effective $c_2,c_3\to r_2,\tilde{r}_3$ scattering. Panel (b) shows that $RT(X_1) \cup RT(X_2)$ can be regarded as HRRT surfaces of $\ew(\tilde{V}_1 \cup V_2)$, where $\tilde{V}_1 \supseteq V_3 \cup X_3 \cup V_1$. We get an effective $\tilde{c}_1,c_2\to r_1, r_2$ scattering. Note that in this case output regions are enlarged although output points $r_1,r_2$ remain unchanged.}
    \label{fig:reduced_3_2}
\end{figure}

\subsubsection{Consequences from pairwise consideration}\label{sec:cons-pairwise}
We now derive further necessary geometric consequences of $\ewV$ being connected by exploiting
pairwise disconnectedness of $\ew(X_i\cup X_j)$ (equivalently $I(X_i:X_j)=0$).
For $n=3$ this pairwise condition admits a direct ``partially connected phase'' interpretation
(see Remark~\ref{rmk:three_case}); for general $n$ it serves as the \emph{base layer} input of a
layered reduction on the boundary lattice.

\begin{remark}\label{rmk:three_case}
For $n=3$, $\ew(V_1\cup V_2\cup V_3)$ is connected if and only if the following hold:
\begin{itemize}
    \item $\sum_{i=1}^{3} |RT(X_i)| \le \sum_{i=1}^{3} |RT(V_i)|$,
    \item $I(X_i:X_j)=0$ for all pairs $i\neq j$.
\end{itemize}
The second item is equivalent to excluding the partially connected phases for $n=3$, so the two items above are complete in this sense.
For $n>3$, the role of the pairwise condition is different: it is the starting point of a layered reduction that ultimately produces an entering ridge at the top layer (see Section~\ref{sec:layered-reduction}).
\end{remark}

\paragraph{Illustration for $n=3$ (one step of the mechanism).}
Consider the boundary causal diamond
\[
\hat{J}^+[\beta_1]\cap \hat{J}^+[\beta_2]\cap
\hat{J}^-[\alpha_1]\cap \hat{J}^-[\alpha_3],
\]
which we denote by $\tilde{Y}_1$
(see Figure~\ref{fig:setup-3})\footnote{
This is analogous to the $2$-to-$(n-1)$ scattering region discussed in
Ref.~\cite{may2022nton}, where it is shown to be nonempty.}.
The past-pointing null sheet emanating from $RT(\tilde{Y}_1)$ intersects $\Sigma_1$
along a curve homologous to $RT(X_1\cup V_2\cup X_2)$.

Suppose that the ridge
\[
\RR_{\tilde{Y}_1,Y_2}
:=\partial J^-[RT(\tilde{Y}_1)]\cap \partial J^-[RT(Y_2)]
\]
lies above
\[
\RR_{X_1,X_2}
:=\partial J^+[RT(X_1)]\cap \partial J^+[RT(X_2)].
\]
Then the bulk geometry takes the form shown in
Figure~\ref{fig:CWT_one_ridge}(c).
Applying the standard focusing argument on null sheets emanating from extremal
surfaces, using $\theta\leq 0$ away from the extremal surface, yields
\begin{equation}
    |RT(X_1)| + |RT(X_2)|
    \;\geq\;
    |RT(V_2)| + |RT(V_1\cup X_1\cup V_2)|.
\end{equation}
This inequality contradicts the assumption that $\ew(X_1\cup X_2)$ is disconnected.
Therefore, the only consistent possibility is that the ridge
$\RR_{\tilde{Y}_1,Y_2}$ lies below $\RR_{X_1,X_2}$, or equivalently,
\begin{equation}\label{eq:connected_cons2_1_old}
    J^+[\ew(X_1)] \cap J^+[\ew(X_2)]
    \cap J^-[\ew(Y_2)] \cap J^-[\ew(\tilde{Y}_1)]
    = \emptyset .
\end{equation}

Furthermore, we can improve \eqref{eq:connected_cons2_1_old} to 
\begin{equation}\label{eq:connected_cons2_1}
    \RR_{\tilde{Y}_1,Y_2}\cap \ew(V_1\cup V_2\cup V_3) \neq \emptyset.
\end{equation}
This follows from the observation that the null sheet $\partial J^-[RT(\tilde{Y}_1)]$ \footnote{We not that by $\partial J^\pm[RT]$, we only refer to the component that points toward the bulk.} lies to the past side of $\partial J^+[RT(X_3)]$. This can be shown using a similar argument as in Lemma \ref{lemma:2_null_ridge} as follows.

First we note that $\partial \hat{J}^+[X_3]\cap \bb=\partial \hat{J}^-[r_3]$ while $\partial \hat{J}^-[\tilde{Y}_1]\cap \bb$ lies strictly to the past of $\partial \hat{J}^-[r_3]$ on $\bb$. That is, the two null sheets, $\partial J^+[RT(X_3)]$ and $\partial J^-[RT(\tilde{Y}_1)]$, do not intersect on $\bb$. Suppose, for contrdition, that in the bulk, $\partial J^-[RT(\tilde{Y}_1)]$  enters the future side of $\partial J^+[RT(X_3)]$. Then the two null sheets must bound a compact set in the bulk or there exists a null generator of $\partial J^-[RT(\tilde{Y}_1)]$ crosses $\partial J^+[RT(X_3)]$ more than once. Such a configuration is incompatible with the achronality of causal boundaries as discussed in Lemma \ref{lemma:2_null_ridge}. This contradiction shows that $\partial J^-[RT(\tilde{Y}_1)]$ lies entirely to the past of $\partial J^+[RT(X_3)]$  \footnote{That they could be tangent does not affect \eqref{eq:connected_cons2_1}.}.

\paragraph{(Optional) scattering interpretation.}
The regions $X_1,X_2$ and $\tilde{Y}_1,Y_2$ may be regarded as the causal complements
of input and output regions in an effective $2$-to-$2$ scattering process
$\tilde{c}_1,c_2\to r_1,r_2$
(see Figure~\ref{fig:reduced_3_2}(b)),
where the \emph{modified input point} $\tilde{c}_1$ is defined, along with the \emph{modified input region} $\tilde{V}_1$, through
$$\hat{J}^+[\tilde{c}_1]\cap \hat{J}^-[r_1]\cap \hat{J}^-[r_2]=\tilde{V}_1= (X_1\cup V_2\cup X_2)^c.$$
The corresponding output regions $\tilde{W}_1,\tilde{W}_2$ are enlarged accordingly, while the output points $r_1,r_2$ remain unchanged.

Note that $\ewW$ being connected does not, by itself, guarantee that
$\ew(\tilde{W}_1\cup \tilde{W}_2)$ is connected.
However, if we additionally assume that $\ew(\tilde{W}_1\cup \tilde{W}_2)$ is connected,
then \eqref{eq:connected_cons2_1_old} can be rephrased as
\begin{equation}\label{eq:connected_cons2_3}
    \ew(\tilde{V}_1\cup V_2)
    \cap
    \ew(\tilde{W}_1\cup \tilde{W}_2)
    \neq \emptyset.
\end{equation}
Accordingly, \eqref{eq:connected_cons2_1} can be rephrased as
\begin{equation}\label{eq:connected_cons2_2}
    \ew(V_1\cup V_2\cup V_3)
    \cap
    \ew(\tilde{W}_1\cup \tilde{W}_2)
    \neq \emptyset.
\end{equation}

\paragraph{Generalization to arbitrary $n$: pairwise entering ridges.}
We now extend the above argument to a general $n$-to-$n$ configuration.

Recall from Section \ref{sec:boundary_setup_n} that boundary left- and right- pointing future-null rays from $c_1, \beta_1,\dots,c_n,\beta_n$ form a lattice on $\bb$. These null rays are also boundary past-null rays emanating from $\alpha_1,r_1,\cdots,\alpha_n,r_n$. 

For any pair $(X_i,X_j),i\neq j$, their complements $(X_i\cup X_j)^c$ consists of two connected components \footnote{Strictly speaker, $\tilde{V}_i$ is the boundary domain of dependence of $V_{j+1}\cup \cdots \cup V_i$. But since we use a domain of dependence and its Cauchy slice interchangablely, we just write as in the text.}
\[\tilde{V}_i=V_{j+1}\cup \cdots \cup V_i \quad \text{and} \quad  \tilde{V}_j=V_{i+1} \cup \cdots \cup V_j,\]
where all indices are understood modulo $n$.
Also recall that points $\beta_i$ and $\beta_j$ as well as antipodal points $r_i$ and $r_j$ are associated with $X_i$ and $X_j$ (see \eqref{eq:X_defn}).

The lattice structure on $\bb$ implies that there exist points $\tilde{\alpha}_i,\tilde{\alpha}_j$ and their past antipodal points $\tilde{c_i}$ and $\tilde{c}_j$ such that their boundary null rays bounds $\tilde{V}_i$ and $\tilde{V}_i$. More precisely, we have
\begin{align}
    \tilde{V}_i = \hat{J}^+[\tilde{c}_i]\cap \hat{J}^-[r_i]\cap \hat{J}^-[r_j],\\
    \tilde{V}_j = \hat{J}^+[\tilde{c}_j]\cap \hat{J}^-[r_i]\cap \hat{J}^-[r_j].   
\end{align}
The two points $\tilde{\alpha}_i,\tilde{\alpha}_j$ also define causal diamonds $\tilde{Y}_i$ and $\tilde{Y}_j$, i.e.
\begin{align}
        \tilde{Y}_i=\hat{J}^-[\tilde{\alpha}_i]\cap \hat{J}^+[\beta_i] \cap \hat{J}^+[\beta_j],\\
    \tilde{Y}_j=\hat{J}^-[\tilde{\alpha}_j]\cap \hat{J}^+[\beta_i] \cap \hat{J}^+[\beta_j]. 
\end{align}

It follows that $\partial J^-[RT(\tilde{Y}_i)]\cap \Sigma_1$ and $\partial J^-[RT(\tilde{Y}_j)]\cap \Sigma_1$ are
homologous to the HRRT surfaces of the two components of $(X_i\cup X_j)^c$, $\tilde{V}_i$ and $\tilde{V}_j$.
Repeating the same geometric construction and focusing calculation, we would get that the ridge $\RR_{\tilde{Y}_i,\tilde{Y}_j}$ lies below $\RR_{X_i,X_j}$, or equivalently,
\begin{equation}\label{eq:connected_cons2_4old}
    \forall\, i\neq j,\qquad
     J^+[\ew(X_i)]\cap J^+[\ew(X_j)] \cap J^-[\ew(\tilde{Y}_i)] \cap J^-[\ew(\tilde{Y}_j)] = \emptyset.
\end{equation}

As in the $3$-party case, the above ordering can be strengthened to show that the ridge must enter $\ewV$, i.e.,
\begin{equation}\label{eq:connected_cons2_4}
    \forall\, i\neq j,\qquad
    \ewV \cap \RR_{\tilde{Y}_i,\tilde{Y}_j} \neq \emptyset .
\end{equation}
This follows from the observations that the ridge $\RR_{\tilde{Y}_i,\tilde{Y}_j}$ lies to the past of all null sheeets $\partial J^+[RT(X_k)], \, \forall k\neq i,j$. The argument proceeds as follows.

The sets $\hat{J}^+[\tilde{V}_i]$ and $\hat{J}^-[\tilde{Y}_i]$ intersect only along their boundaries, since $\tilde{c}_i$ and $\tilde{\alpha}_i$ are antipodal points. In particular, their intersection has empty interior. Also note that $X_k\subset \tilde{V}_i$ for $j+1\leq k \leq i-1$. Therefore, for $j+1\leq k \leq i-1$, $\partial J^+[RT(X_k)],$ and $\partial J^-[RT(\tilde{Y}_i)]$ do not intersect on $\bb$; more precisely, $\partial J^-[RT(\tilde{Y}_i)]\cap \bb$ lies to past side of $\partial J^+[RT(X_k)]\cap \bb$. By a similar argument as in Lemma \ref{lemma:2_null_ridge}, $\partial J^-[RT(\tilde{Y}_i)]$ also lies to the past of $\partial J^+[RT(X_k)]$ in the bulk, for $j+1\leq k \leq i-1$.

For $i+1\leq k \leq j-1$, we have $X_k\subset \tilde{V}_j$. A similar arugment yields that $\partial J^-[RT(\tilde{Y}_j)]$ lies to the past of $\partial J^+[X_k]$ for such $k$. 

Combining the two arguments, we get that the ridge  $\RR_{\tilde{Y}_i,\tilde{Y}_j}=\partial J^-[RT(\tilde{Y}_i)]\cap\partial J^-[RT(\tilde{Y}_j)]$ lies to the past of all null sheets $\partial J^+[RT(X_k)], \, \forall k\neq i,j$. This observation together with \eqref{eq:connected_cons2_4old} leads to \eqref{eq:connected_cons2_4}, since $\ewV=\big(\ew(X_i^c)\cap \ew(X_j^c) \big) \, \cap \, \big( \bigcap_{k\neq i,j}\ew(X_k^c)$ \big).

\paragraph{(Optional) scattering interpretation.}
Similar to the $n=3$ case, we have an effective $2$-to-$2$ scattering process for an aribitrary pair $i\neq j$: \[\tilde{c}_i, \tilde{c}_j \to r_i, r_j.\]
Let us also modified output region accordingly: 
\begin{align*}
    \tilde{W}_i=J^-[r_i]\cap J^+[\tilde{c}_i]\cap J^+[\tilde{c}_j],\\
    \tilde{W}_j=J^-[r_j]\cap J^+[\tilde{c}_i]\cap J^+[\tilde{c}_j].    
\end{align*}

If we additionally assume that $\ew(\tilde{W}_i\cup \tilde{W}_j)$ is connected, then \eqref{eq:connected_cons2_4old} can be rephrased as
\begin{equation}
    \ew(\tilde{V}_i\cup \tilde{V}_j)\cap \ew(\tilde{W}_i \cup \tilde{W}_j) \neq \emptyset.
\end{equation}
Accordingly, \eqref{eq:connected_cons2_4} can be rephrased as 
\begin{equation}
\ewV \cap \ew(\tilde{W}_i \cup \tilde{W}_j) \neq \emptyset.    
\end{equation}

\subsubsection{Layered consequences from successive regroupings}
\label{sec:layered-reduction}
We now organize the general-$n$ consequences into a layered reduction on the boundary lattice. At each layer one
regroups the boundary into an effective \(n'\)-to-\(n'\) configuration and applies the top-layer entering-ridge argument to that effective configuration.

More precisely, since \(\ewV\) is connected, Lemma~\ref{lemma:disconnect_lemma} implies that every subcollection of the complementary regions \(X_i\) has fully disconnected entanglement wedge. Let
\[
A\subset \{1,\ldots,n\},\qquad |A|=m\ge 2 .
\]
The boundary complement
\[
\left(\bigcup_{a\in A} X_a\right)^c
\]
decomposes, with respect to the cyclic boundary lattice, into \(m\) connected effective input regions
\[
\widetilde V^{(A)}_1,\ldots,\widetilde V^{(A)}_{m} .
\]
Equivalently, the elements of \(A\) specify the \(m\) remaining gaps in the boundary lattice, while the omitted \(X\)'s are absorbed into the effective input regions \(\widetilde V^{(A)}_\mu\). For this effective \(m\)-party configuration, define the corresponding
auxiliary output diamonds
\[
\widetilde Y^{(A)}_1,\ldots,\widetilde Y^{(A)}_{m}
\]
by the same antipodal-point construction used for the original \(Y_i\)'s.

We organize the consequences by increasing subset size 
\[m=2,\ldots,n.\]
At level \(m\), we consider all subsets 
\[A\subset \{1,\dots,n\}, \qquad |A|=m.\]
Since the construction proceeds in increasing order of \(m\), all statements associated with proper subsets \(B\subset A\) have already appear at eariler levels, i.e. smaller $m$. Therefore, for a subset \(A\) of size $m$, the only genuinely new statement is the top-layer consequence of the effective $m$-party configuration determined by $A$.

Applying the top-layer connected--disconnected comparison of Section \ref{sec:cons-conn-disc} to the $m$-level effective configuration gives
\[
\exists\, \mu\neq \nu
\quad\text{such that}\quad
\ew(\tilde{V}_1^{(A)}\cup \cdots \cup \tilde{V}_{m}^{(A)}) \cap
\RR_{\widetilde Y^{(A)}_\mu,\widetilde Y^{(A)}_\nu}
\neq \emptyset.
\]
This is the entering-ridge consequence at layer \(m\). Note that we do not directly get $\ewV \cap \RR_{\widetilde Y^{(A)}_\mu,\widetilde Y^{(A)}_\nu} \neq \emptyset$.

To understand the difference, it is useful to classify the original complementary regions $X_i$ into three types relative to a given effective configuration and a chosen auxiliary ridge $\RR_{\widetilde Y_k,\widetilde Y_l}$:
\begin{itemize}

\item[(I)] \textbf{Effective complement regions.}
These are the $m$ regions $X_{i_1},\ldots,X_{i_{m}}$ that remain as the complement gaps in the effective $m$-party problem. If the ridge enters the effective wedge, it is controlled with respect to these $X_{i_a}$, i.e.\ it lies inside
\[\bigcap_{a=1}^{m} \ew(X_{i_a}^c).\]

\item[(II)] \textbf{Absorbed regions in the selected blocks.}
These are the $X_i$ contained inside the two effective input regions $\widetilde V_k$ and $\widetilde V_l$ whose associated auxiliary diamonds form the ridge. For these $X_i$, the boundary ordering argument of Lemma~\ref{lemma:2_null_ridge} gives a one-sided control, so that the ridge lies to the past of the corresponding null sheets \footnote{We note that the boudnary ordering argument of Lemma 2.1 does not apply to the case that the two causal boundaries share a null ray on $\bb$. For example, we cannot conclude that $\partial J^-[RT(\tilde{Y}_1)]$ lies to the past of $\partial J^-[RT(Y_1)]$.}.

\item[(III)] \textbf{Remaining absorbed regions.}
These are the $X_i$ absorbed into the other effective input regions $\widetilde V^{(A)}_\mu$ with $\mu\neq k,l$. For these regions, neither mechanism applies, and no ordering relative to the corresponding null sheets is available.
\end{itemize}

\medskip

This classification explains the special role of the endpoint cases.

\paragraph{Base layer ($m=2$).}
In this case, there are only two effective input regions
\[
\widetilde V_k,\quad \widetilde V_l.
\]
All $X_i$ fall into Type~(I) or Type~(II), and there are no Type~(III) regions. Consequently, the auxiliary ridge satisfies the full set of ordering conditions and lies inside
\[
\ewV=\bigcap_{i=1}^n \ew(X_i^c).
\]
This is precisely the pairwise entering-ridge statement of Section~\ref{sec:cons-pairwise}.

\paragraph{Top layer ($m=n$).}
In this case, no regrouping occurs and all $X_i$ are of Type~(I). The top-layer argument directly produces a ridge $\RR_{Y_k,Y_l}$ entering $\ewV$.

\paragraph{Intermediate layers ($2<m<n$).}
In these cases, Type~(III) regions are present. Although the auxiliary ridge enters the effective entanglement wedge, the lack of control with  respect to the Type~(III) regions prevents one from concluding that it lies inside $\ewV$. Thus intermediate layers do not, by themselves, produce additional entering-ridge constraints on $\ewV$. 

However, for certain choices of the auxiliary ridge
$\RR_{\widetilde Y_k,\widetilde Y_l}$, it may happen that all $X_i$ not in the effective complement are absorbed into the two selected blocks $\widetilde V_k$ and $\widetilde V_l$, so that no Type~(III) regions are present. In such cases, the corresponding ridge does enter $\ewV$.

\medskip

We summarize discussions in this section in the following theorem.
\begin{theorem}[Entering-ridge consequences of connected \(\ewV\)]
\label{thm:sufficient_n}
Assume Assumption~\ref{assumption:1}. If \(\ewV\) is connected, then:

\begin{enumerate}
\item \textbf{Top layer.}
There exist a pair \(k\neq l\) such that
\[
\ewV\cap \RR_{Y_k,Y_l}\neq \emptyset,
\qquad
\RR_{Y_k,Y_l}:=
\partial J^-[RT(Y_k)]\cap \partial J^-[RT(Y_l)].
\]

\item \textbf{Pairwise/base layer.}
For every pair \(i\neq j\), let
\[
(X_i\cup X_j)^c=\widetilde V_i\cup \widetilde V_j
\]
be the decomposition into two connected components on the boundary
lattice, and let \(\widetilde Y_i,\widetilde Y_j\) be the corresponding
auxiliary output diamonds. Then
\[
\ewV\cap \RR_{\widetilde Y_i,\widetilde Y_j}\neq \emptyset, \quad \forall i\neq j.
\]

\item \textbf{Intermediate layers.}
More generally, for every subset
\[
A\subset\{1,\ldots,n\},\qquad |A|=m\ge 2,
\]
let
\[
\left(\bigcup_{a\in A}X_a\right)^c=
\widetilde V^{(A)}_1\cup\cdots\cup \widetilde V^{(A)}_{m}
\]
be its decomposition into effective input regions, and let
\(\widetilde Y^{(A)}_1,\ldots,\widetilde Y^{(A)}_{m}\) 
be the associated auxiliary output diamonds. Then there exist \(\mu\neq\nu\) such that
\[
\ew(\widetilde V^{(A)}_1\cup\cdots\cup \widetilde V^{(A)}_{m}) \cap
\RR_{\widetilde Y^{(A)}_\mu,\widetilde Y^{(A)}_\nu}
\neq \emptyset .
\]
\end{enumerate}

\end{theorem}

\begin{remark}
For $n=3$, connectedness of $\ew(V_1\cup V_2\cup V_3)$ is equivalent to the two conditions in Remark~\ref{rmk:three_case}. In particular, the pairwise condition $I(X_i:X_j)=0$ admits a direct interpretation in terms of excluding partially connected phases.

For $n>3$, the situation is intrinsically multipartite: connectedness of $\ewV$ imposes constraints not captured by pairwise statements alone. In the ridge language, this multipartite structure is  naturally organized by the layered reduction of Section~\ref{sec:layered-reduction}, which iteratively propagates ``entering ridge'' statements from the top-level ridge $\RR_{Y_k,Y_l}$ entering $\ewV$ to auxiliary diamonds.
\end{remark}

\subsection{Generalized bulk scattering regions}\label{sec:S_E_condition}
In the $2$-to-$2$ asymptotic sattering problem, a generalized bulk scattering region
\begin{equation}
    \mathcal{S}_E = \ew(V_1\cup V_2) \cap \ew(W_1\cup W_2)
\end{equation}
or the modified version
\begin{equation}
    \tilde{\mathcal{S}}_E = \big[\ew(V_1\cup V_2)/(\ew(V_1)\cup \ew(V_2))\big] \bigcap \big[\ew(W_1\cup W_2)/(\ew(W_1)\cup \ew(W_2))\big]
\end{equation}
was identified to characterize connectedness of $\ewV$ and $\ewW$. 
Here we trivially generalize the above bulk scattering region as
\begin{equation}
    \mathcal{S}_E = \ewV \cap \ewW
\end{equation}
and discuss necessary conditions for $\mathcal{S}_E\neq \emptyset$.

\begin{theorem}\label{thm:S_E_condition}
Assume the standard conditions listed in Assumption \ref{assumption:1}. Suppose that both $\ewV$ and $\ewW$ are connected.

    If, in addition, for every pair $i\neq j$, the ridge $\mathcal{R}_{Y_i,Y_j}$ enters $\ewV$, i.e.,
    \begin{equation}\label{eq:SE_nonempty}
        \ewV \cap \ew(Y_i^c)\cap \ew(Y_j^c)\neq \emptyset, \quad \forall \, i\neq j,
    \end{equation}
    then the generalized scattering region is nonempty:
    $$\mathcal{S}_E \neq \emptyset.$$

Equivalently, condition~\eqref{eq:SE_nonempty} can be expressed in terms of the enlarged output regions
\[
\tilde W_i\cup \tilde W_j = (Y_i\cup Y_j)^c,
\]
as
\begin{equation}
    \ewV \cap \ew(\tilde W_i\cup \tilde W_j)\neq \emptyset,
    \qquad \forall\, i\neq j.
\end{equation}
\end{theorem}

\begin{remark}
The theorem is a Helly‑type result: pairwise intersections among imply a common intersection. The convexity required in classical Helly theorems is replaced here by preperties, e.g. the achronality, of causal boundaries.
\end{remark}
\begin{remark}
    One can use focusing calculations, which should be very familiar by now, to show that if the input regions $V_i$'s have pairwise connected entanglement wedges, i.e. $I(V_i:V_j)>0, \forall i\neq j$, then \eqref{eq:SE_nonempty} is necessarily satisfied.
\end{remark}

\begin{proof}
    We first note that $\ewV \cap \ewW\neq \emptyset$ is equivalent to the future horizon of $\ewV$ and the past horizon of $\ewW$ intersects because both $\ewV$ and $\ewW$ are compact sets in $\overline{M}$.

    We use a similar geometric structure as before.
    Consider the upper horizon of $\ewV$ or $\ZZi$ in the notation of section \ref{sec:connect_cons}. As argued above, $\NN_{Y_i}=\partial J^-[\ew(Y_i)]$ intersect $\ZZi$ at a simple curve $\CC_i$, with endpoints of $\partial V_i$, i.e. $a_i$ and $b_i$. Moreover, $\CC_i$ together with the boundary future horizon $\hat{H}^+[V_i]$ of $V_i$ bounds a compact set $D_i$ on $\ZZi$ (For an illustration see Figure \ref{fig:ZZin} with $X$'s and $V$'s switched). We have the following two claims about $D_i$'s.

\noindent\textbf{Claim 1.} Each $D_i$ is topologically a disk.

\smallskip

\noindent\emph{Proof of Claim 1.}
    First, the boundary future horizon $\hat{H}^+[V_i]$ lies on $\bb$ while the interior of $\CC_i$ lies in the bulk. Therefore, the two curves meet only at their common endpoints $a_i,b_i$ and do not intersect elsewhere. It follows that 
    \[\CC_i \cup \hat{H}^+[V_i]\]
    is a simple closed curve on $\ZZi$.
      
The hypersurface $\ZZi$ is constructed from null sheets emanating from HRRT surfaces and glued along codimension-two ridges. In particular, $\ZZi$  is a disk with boundary \(\hat H^+[V_1]\cup RT(X_1)\cup  \cdots\cup \hat H^+[V_n]\cup RT(X_n)\). By the Jordan curve theorem on a disk, the simple closed curve \(\CC_i\cup \hat H^+[V_i]\) bounds a subdisk of \(\ZZi\). This subdisk is \(D_i\).
\hfill $\triangle$

\medskip
\noindent\textbf{Claim 2.} The regions $D_i$ are cyclically ordered.

\smallskip

\noindent\emph{Proof of Claim 2.}
This follows directly from the configuration on $\bb$: since the causal diamonds $V_i$ are cyclically ordered on $\bb$, their future horizons $\hat{H}^+[V_i] \subset \bb$ inherit the same cyclic ordering. By definition, each $D_i$ is attached to the corresponding boundary $\hat{H}^+[V_i]$. Therefore, the disks $D_i$ are arranged in the same cyclic order on $\ZZi$.
\hfill $\triangle$

\medskip 
    
    The condition \eqref{eq:SE_nonempty} is equivalent to the following condition:
    \begin{equation}\label{eq:D_pairwise}
        D_i\cap D_j\neq \emptyset,\qquad \forall i\neq j.
    \end{equation}
    Thus it suffices to show that the pairwise intersections of the $D_i$ force a common point.
    
    \medskip
    \noindent\textbf{Graph/Tree of null generators.}
    The surface $\ZZi$ is foliated by future‑directed null geodesics---the generators of the sheets $\mathcal{N}_{X_i}=\partial J^+[\mathrm{RT}(X_i)]$.
    These generators either end on a ridge $\mathcal{R}_{X_i,X_j}$ (where two sheets meet) or on a boundary arc of $\hat H^+[V_k]$. For each $k$, the two boundary arcs of $\hat H^+[V_k]$ meet at the future tip $d_k$.
    
    Define the set
    \[
    R \;=\; \bigl(\text{union of all ridge segments in } \ZZi\bigr)\;\cup\;\{d_1,\dots,d_n\}.
    \]
    Construct a continuous retraction $r:\ZZi\to R$ as follows.
    If a point lies on a generator that ends at a ridge point, map it to that ridge point.
    If a point lies on a generator that ends on $\hat H^+[V_k]$, map it to the tip $d_k$.
    On the boundary arcs themselves the map sends every point to the corresponding tip; this is consistent because the boundary arcs consist of generators that end on the tip.
    The map $r$ is continuous because the family of generators varies continuously and the assignment of tips is locally constant along the boundary.
    
    \noindent\textbf{$R$ is a finite tree.}
    The ridges are simple curves that intersect only at ridge-junctions where more than three null sheets meet (Corollary~\ref{lemma:3_ridge}); they contain no loops.
    Each ridge either ends at a tip $d_i$ on $\partial M$ or becomes an edge connecting two junctions.
    The resulting space $R$ is a one‑dimensional compact graph.

    Because $\ZZi$ is a topological disk (Claim~1), $R$ is simply connected and therefore a finite tree.
    Its leaves are precisely the tips $d_1,\dots,d_n$; its interior vertices are the ridge‑junctions.

    \noindent\textbf{The $D_i$ project to subtrees.}
    The curve $\CC_i$ lies in the achronal past lightsheet $\NN_{Y_i}=\partial J^-[\mathrm{RT}(Y_i)]$; hence any null generator of $\ZZi$ can intersect $\CC_i$ at most once.
    For a generator that intersects $\CC_i$ and ends on the boundary arc $\hat H^+[V_i]$, the segment between $\CC_i$ and the boundary $\bb$ lies inside $D_i$; under the retraction this segment maps into the tip $d_i$.
    For a generator that intersects $\CC_i$ and ends on a ridge point, the segment between $\CC_i$ and the ridge lies inside $D_i$; under the retraction this segment maps into the corresponding ridge point.
    Hence the retracted image
    \[
    I_i := r(D_i) \;\subset\; R
    \]
    is connected and contains the leaf $d_i$ (the entire boundary arc $\hat H^+[V_i]$ maps to $d_i$).
    
    Moreover, $I_i$ contains no other leaf $d_k$ ($k\neq i$): if a generator ended at $d_k$ and carried a point of $D_i$, that generator would have to enter $D_i$ through $\CC_i$ and exit again to reach the boundary at $d_k$, intersecting $\CC_i$ twice---violating the achronality of $\NN_{Y_i}$.

    A connected subset of a tree is itself a subtree; therefore each $I_i$ is a subtree of $R$.

    \noindent\textbf{Helly property for trees.}
    Equation~\eqref{eq:D_pairwise} implies $I_i\cap I_j\neq\varnothing$ for all $i\neq j$.
    The \emph{Helly property for trees} (see Appendix \ref{sec:app} for a statement) states that any finite family of subtrees of a tree that pairwise intersect has a common vertex.
    Hence there exists a vertex $v\in R$ with $v\in\bigcap_{i=1}^{n} I_i$.
    
    \noindent\textbf{The common vertex is a bulk ridge‑triple point.}
    Each $I_i$ is attached to a distinct leaf $d_i$.
    Consequently $v$ cannot be a leaf, for it would then have to equal every leaf simultaneously.
    Hence $v$ is an interior vertex of $R$, i.e.\ a point where three or more ridges meet---a bulk ridge‑junction point $p$.
    
    \noindent\textbf{The ridge point lies in all $D_i$.}
    Since $p = v \in I_i = r(D_i)$, there exists a point $q_i\in D_i$ with $r(q_i)=p$ for all $i$.
    The generator through $q_i$ ends at $p$, and because a generator can intersect $\CC_i$ at most once, the segment of that generator between $q_i$ and $p$ lies entirely in $D_i$.
    As $D_i$ is compact, $p$ (the endpoint of that segment) belongs to $D_i$.
    This holds for every $i$, so
    \[
    p \;\in\; \bigcap_{i=1}^{n} D_i \;\neq\; \varnothing ,
    \]
    which is exactly $\mathcal{S}_E\neq\varnothing$.
\end{proof}

 \subsubsection{Interpretation and optimality of the condition.}\label{sec:S_E_nonempty_discussion}
Condition~\eqref{eq:SE_nonempty} can be expressed in terms of the enlarged
output regions $\tilde W_i\cup\tilde W_j = (Y_i\cup Y_j)^c$ as
\[
    \ewV \cap \ew(\tilde W_i\cup\tilde W_j)\neq\varnothing ,\qquad \forall\, i\neq j .
\]
For $n=2$ there is only one such pair, and the condition reduces to 
$$\ew(V_1\cup V_2)\cap\ew(W_1\cup W_2)\neq\varnothing.$$
By Theorem~\ref{thm:sufficient_n}, this is automatically guaranteed when both $\ew(V_1\cup V_2)$ and $\ew(W_1\cup W_2)$ are connected, so Theorem~\ref{thm:S_E_condition} recovers the known result that connected input and output wedges imply $\mathcal{S}_E\neq\emptyset$ \cite{zhao2025proof,lima2025sufficientGCWT}.  (The full equivalence between wedge connectivity and a non‑empty intersection region holds for the modified region in which the individual wedges are subtracted~\cite{zhao2025proof,lima2025sufficientGCWT}.)
For $n>2$, however, it becomes a genuinely multipartite constraint: the bulk must contain a common ``meeting region'' for all pairs of modified output wedges inside the input wedge.

The strength of this condition is best seen by comparing it with the consequences of wedge connectivity.
Theorem~\ref{thm:sufficient_n} shows that when the input wedge is connected, there exists at least one complementary ridge
$\RR_{Y_k,Y_l}$ that lies below all $X$-sheets (i.e.\ enters
$\ewV$).  By time reflection, connectedness of the output wedge
guarantees at least one ridge $\RR_{X_i,X_j}$ that lies above all $Y$-sheets.  Condition~\eqref{eq:SE_nonempty}, on the other hand, demands that \emph{every} ridge $\RR_{Y_i,Y_j}$ lies below the $X$-sheets---equivalently, all $Y$-ridges are ``lower'' than all $X$-ridges.  Thus, our condition for $\mathcal{S}_E\neq\varnothing$ is strictly stronger
than the simultaneous connectedness of the input and output wedges.: while a single ridge from each family suffices to enforce connectedness the two wedges, a common intersection region exists only when the whole family of past‑pointing ridges passes through the future horizon of the input wedge.

We have not proved that~\eqref{eq:SE_nonempty} is necessary for
$\mathcal{S}_E\neq\varnothing$; the condition seems difficult to relax significantly.  
Proving that no weaker condition exists would amount to a full characterisation of $\mathcal{S}_E\neq\varnothing$, which we leave as an open problem.
Constructing explicit examples with connected wedges but $\mathcal{S}_E=\varnothing$ would be an important step in that
direction; we expect that a systematic numerical investigation,
e.g.\ using matter perturbations in pure $\mathrm{AdS}_3$ with
$n$-to-$n$ configurations, can produce such examples.  This would also clarify the gap between mere pairwise connectivity 
(governed by pairwise mutual information) and the genuinely multipartite  entanglement structure that is constrained by the monogamy of mutual information and the no‑GHZ result~\cite{holography-no-GHZ}.

\section{Conclusion and Discussion}\label{sec:discussion}
Similar to the proofs of CWT and $n$-to-$n$, the proofs given here will also work for semiclassical spacetimes that satisfy the quantum maximin formula \cite{akers2020quantummaxmin} and the quantum focusing conjecture \cite{bousso2016quantumfocusing}.

The configuration of disjoint input regions $\Vall$ fully specifies the boundary setup for the asymptotic $n$-to-$n$ scattering problem. Consequently, the results derived here apply generally to the entanglement wedge structure of multipartite, spacelike-separated boundary regions with shape of causal domains (not allowing adjacent regions).

\subsection{Two kinds of scattering regions.}\label{sec:comparison_Null_sheet}
We make an observation that would help to understand relations among different generalizations of the $2$-to-$2$ and $n$-to-$n$ connected wedge theorems.

We noted that there are two null sheets anchored to $\hat{J}^+[c_i]$ appeared in the proofs: the future-pointing null sheet $\NN_{V_i}$  emanating from $RT(V_i)$ and the past-pointing null sheet $\NN_{Y_i}$ emanating from $RT(Y_i)$. 
Similarly, for $\hat{J}^+[\beta_i]$, we would consider the future-pointing null sheet $\mathcal{N}_{X_i}$ emanating from $RT(X_i)$ and the past-pointing null sheet $\mathcal{N}_{W_i}$ emanating from $RT(W_i)$.

One can argue for an positioning relation between $\mathcal{N}_{V_i}$ and $\mathcal{N}_{Y_i}$ with the same $i$. See Figure \ref{fig:NV_NY} for an illustration. First note that the future pointing null sheet $\mathcal{N}_{V_i}$ from $RT(V_i)$ is the future horizon of $\ew(V_i^c)$. On the boundary $\bb$, $\hat{D}[V_i^c]$ contains $Y_i$ (since $Y_i\subseteq \hat{J}^-[\alpha_i]$ and $\alpha_i$ is antipodal point of $c_i$), thus $\ew(Y_i) \subseteq \ew(V_i^c)$ by entanglement wedge nesting. This implies that $RT(Y_i)$ is spacelike separated from $\NN_{V_i}$ and moreover, $RT(Y_i)$ lies closer to $Y_i$ than $\NN_{V_i}\cap \Sigma_2$. Similarly, the past null sheet $\NN_{Y_i}$ is the past horizon of $\ew(Y_i^c)$ and $\ew(Y_i^c) \supseteq \ew(V_i)$ . This implies that $RT(V_i)$ is spacelike separated from $\NN_{Y_i}$ and moreover, $RT(V_i)$ lies closer to $V_i$ than $\NN_{Y_i} \cap \Sigma_1$. To sum up, we have two null sheets, $\mathcal{N}_{V_i}$ and $\mathcal{N}_{Y_i}$,  coinciding on $\bb$, with one lyinng to a specific side relative to the other at $\Sigma_1$ and $\Sigma_2$: $\NN_{V_i}$ lies closer to $V_i$ (or closer to $\bb$) than $\NN_{Y_i}$ on $\Sigma_1$ while $\NN_{Y_i}$ lies closer to $Y_i$  (or closer to $\bb$) than $\NN_{V_i}$ on $\Sigma_2$.  If they intersect in the bulk and do not coincide everywhere, then their intersection would necessarily contain either multiple connected components or a closed loop. Both possibilities are excluded by the same separation and generator-uniqueness argument used in Lemma~\ref{lemma:2_null_ridge}.
To avoid contradiction with causality, they have empty intersection or coincide completely. One can also derive this fact from the maximum principle using a similar argument as in refs. \cite{galloway1999maximum} or \cite{wall2014maximin}.
\begin{figure}
    \centering
    \includegraphics[width=0.3\linewidth,trim={10 5 10 5},clip]{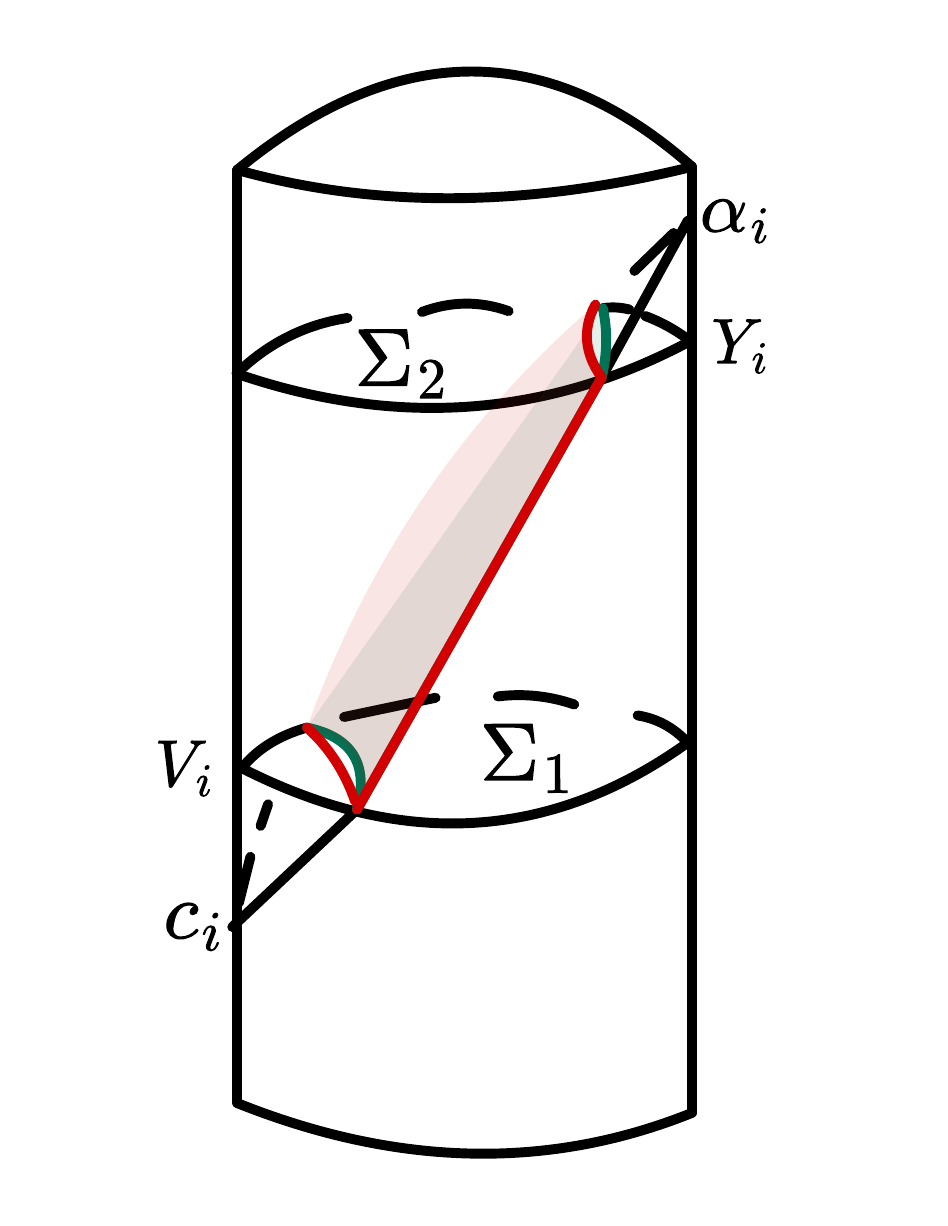}
    \caption{Illustration of position relation between $\NN_{V_i}$ and $\NN_{Y_i}$. The red sheet or curves are associated with $\NN_{V_i}=\partial J^+[RT(V_i)]$. The green sheet or curves are associated with $\NN_{Y_i}=\partial J^-[RT(Y_i)]$. The points $c_i$ and $\alpha_i$ are antipodal points.}
    \label{fig:NV_NY}
\end{figure}
A similar reasoning would show that $\NN_{X_i}$ and $\NN_{W_i}$ either coincide completely or have empty intersection. Specifically, $\NN_{X_i}$ lies closer to $X_i$ than $\NN_{W_i}$ (or equivalently $\NN_{W_i}$ lies closer to $W_i$ than $\NN_{X_i}$).

In the generalized $2$-to-$2$ Connected Wedge Theorem by ref. \cite{may2021region}, a bulk region 
\begin{equation}
    \mathcal{S}'_E = J^+[\ew(V_1)] \cap J^+[\ew(V_2)] \cap J^-[\ew(W_1)] \cap J^-[\ew(W_2)]
\end{equation}
is used, whose nonemptiness was shown to imply connectedness of $\ew(V_1\cup V_2)$. 
On the other hand, in the generalized $2$-to-$2$ Connected Wedge Theorem by refs. \cite{zhao2025proof,LL2025superadditivity,lima2025sufficientGCWT}, another bulk region 
\begin{equation}
    \mathcal{S}_E = \ew(V_1\cup V_2)\cap \ew(W_1\cup W_2)
\end{equation}
is used, whose nonemptiness was shown to follow from connectedness of $\ew(V_1\cup V_2)$ and $ \ew(W_1\cup W_2)$. 

The bulk region $S'_E$ involves $\NN_{V_i}$ and $\NN_{W_i}$ while the bulk region $\mathcal{S}_E$ involves $\NN_{Y_i}$ and $\NN_{X_i}$. The above observation on the inclusion relation between $\NN_{V_i}$ and $\NN_{Y_i}$ or between $\NN_{X_i}$ and $\NN_{W_i}$ would immediately imply that 
\begin{equation}
    \mathcal{S}_E'\subseteq \mathcal{S}_E,
\end{equation}
as one would expect from the two generalization of the $2$-to-$2$ Connected Wedge Theorem.

A similar remark applies to the $n$-to-$n$ scattering problem discussed here. The bulk region
    $$J^+[\ew(V_i)]\,\cap\, J^+[\ew(V_j)] \, \cap\, J^-[\ew(W_k)] \cap\, J^-[\ew(W_l)], \quad,i\neq j, k\neq l$$
involved in the necessary condition (Theorem \ref{thm:weaker-necessary}) is contained in the bulk region 
$$    J^-[RT(Y_k)]\cap J^-[RT(Y_l)] \cap J^+[\ew(X_i)] \cap J^+[\ew(X_j)]
$$
in the sufficient condition (Theorem \ref{thm:sufficient_n}).

\subsubsection{Comparison with the superadditivity proposal.}
Leutheusser and Liu~\cite{LL2025superadditivity} interpreted the geometric connectedness of entanglement wedges as the signal of superadditivity of bulk operator algebras, which provides the resource for the corresponding boundary quantum task.

For $n=2$ this equivalence has been proved~\cite{zhao2025proof,lima2025sufficientGCWT}.
Our results for $n>2$ show that the geometric condition required for $\mathcal{S}_E\neq\emptyset$ seems strictly stronger than simultaneous connectedness of the input and output wedges (see Section \ref{sec:S_E_nonempty_discussion}).
Thus a simple one‑to‑one correspondence between a single ``bridging operator'' and the geometric intersection region seems too naive for $n>2$; the full pairwise superadditivity structure among the output algebras appears to be needed.
It is conceivable that the layered reduction of
Theorem~\ref{thm:sufficient_n} could strengthen the entering‑ridge conditions to narrow the gap between the single ridge that follows from wedge connectivity and the full set of ridges required for $\mathcal{S}_E\neq\emptyset$, thereby bringing us closer to an equivalence; we leave a complete characterization as an open problem.

\subsection{Future Directions}
For multipartite $n>2$ scattering processes, the holographic dictionary appears less transparent than in the $n=2$ case. In particular, the condition $\mathcal{S}_E\neq \emptyset$ seems more restrictive than the mere connectedness of $\ewV$ and $\ewW$. Our analysis addresses this increased complexity by reducing the problem to a pairwise framework and yields several concrete geometric consequences. This reduction is justified by the intrinsic structure of holographic states, which are highly constrained and, in particular, cannot support purely GHZ-like multipartite entanglement \cite{holography-no-GHZ}.

While the results presented here do not yet constitute a complete or equivalent geometric characterization for $n>2$, they significantly extend the core principle established for $2$-to-$2$ scattering: nontrivial boundary quantum protocols are faithfully encoded in specific, nonlocal geometric signatures in the bulk spacetime. A fully equivalent characterization for general $n$ will likely require genuinely multipartite information-theoretic tools. In this regard, Theorem~\ref{thm:NMI-wedge} provides a complete information-theoretic characterization of when the entanglement wedge of multiple disjoint spacelike-separated regions is connected. We leave further exploration along these lines to future work.

In this paper, we focused on identifying necessary as well as sufficient conditions for the connectedness of multipartite entanglement wedges. An alternative and complementary perspective would be to formulate necessary or sufficient conditions directly for the existence of bulk-only scattering processes \cite{maldacena2017bulkpointsingularity}, namely configurations satisfying
\begin{equation}
    \bigcap_i J^+[c_i] \,\cap\, \bigcap_j J^-[r_j] \neq \emptyset
    \quad \text{while} \quad
    \bigcap_i \hat{J}^+[c_i] \,\cap\, \bigcap_j \hat{J}^-[r_j] = \emptyset.
\end{equation}
Finally, it would also be interesting to generalize the discussions in this work to higher-dimensional asymptotically AdS spacetimes.

\acknowledgments
I thank Edward Witten for getting me interested in this subject and Shing-Tung Yau for his enduring support. I also thank the anonymous referee for helpful suggestions to improve the manuscript.

\appendix

\section{Helly property for trees}\label{sec:app}

A family of subsets of a set is said to have the \emph{Helly property} if,
whenever every two members of the family intersect, the whole family has
a non‑empty intersection.
For intervals on the real line this is the classical Helly theorem (Helly~1913).
The same property holds for subtrees of a tree (see e.g.~\cite{bondy1976graph}):
\begin{lemma}[Helly property for trees]\label{lem:helly_tree}
Let $T$ be a finite tree and let $T_1,\dots,T_n$ be subtrees of $T$.
If $T_i\cap T_j\neq\varnothing$ for all $i\neq j$, then
$\bigcap_{i=1}^{n} T_i \neq \varnothing$.
\end{lemma}
In fact one can prove that the intersection is even a subtree, but the
vertex‑intersection form is sufficient for our purposes.
This result is a standard fact in combinatorial geometry and is the
precise substitute for ordinary convexity that appears in the proof of
Theorem~\ref{thm:S_E_condition}.

\bibliographystyle{JHEP}
\bibliography{Literature.bib}

\end{document}